\documentclass[
  aps,
  pra,
  10pt,
  twocolumn,
  superscriptaddress,
  nofootinbib
]{revtex4-2}

\usepackage{amsmath,amssymb,amsthm}
\usepackage{mathtools}

\usepackage{braket}
\usepackage{empheq}

\usepackage{silence}
\usepackage{hyperref}
\usepackage{cleveref}

\usepackage{enumitem}

\newcommand{\eval}[1]{\left.#1\right\vert}
\newcommand{\inp}[2]{\langle #1, #2 \rangle}
\DeclareMathOperator{\im}{im}
\DeclareMathOperator{\wt}{wt}
\DeclareMathOperator{\spn}{span}
\DeclareMathOperator{\rank}{rank}
\DeclareMathOperator{\Hom}{Hom}

\theoremstyle{plain}
\newtheorem{theorem}{Theorem}[section]
\newtheorem{proposition}{Proposition}[section]
\newtheorem{corollary}{Corollary}[section]
\newtheorem{lemma}{Lemma}[section]

\theoremstyle{definition}
\newtheorem{definition}{Definition}[section]

\theoremstyle{remark}
\newtheorem{remark}{Remark}[section]

\Crefname{theorem}{Theorem}{Theorems}
\Crefname{proposition}{Proposition}{Propositions}
\Crefname{corollary}{Corollary}{Corollaries}
\Crefname{definition}{Definition}{Definitions}
\Crefname{lemma}{Lemma}{Lemmas}
\Crefname{remark}{Remark}{Remarks}
\Crefname{equation}{Eq.}{Eqs.}
\Crefname{appendix}{Appendix}{Appendices}

\begin{document}

\title{
  Homological origin of the transversal implementability of logical diagonal gates in quantum CSS codes
}

\author{Junichi Haruna}
\email{j.haruna1111@gmail.com}
\affiliation{Graduate School of Informatics, Kyoto University, Kyoto 606-8501, Japan}
\affiliation{Center for Quantum Information and Quantum Biology, University of Osaka, Toyonaka, Osaka 560-0043, Japan}


\begin{abstract}
Transversal Pauli $Z$ rotations provide a natural route to fault-tolerant logical diagonal gates in quantum Calderbank--Shor--Steane (CSS) codes, but their capability is inherently constrained. 
We develop a homological framework that organizes transversal diagonal gates in terms of their logical action and physical implementation, revealing two layers of structure that govern their behavior.
At a fixed level, we establish that their logical action admits a classification in terms of homological data of the underlying chain complex, extending the standard description of logical operators. 
We then formulate the refinement to finer angles as a lifting problem and derive two Bockstein-type obstruction maps, whose vanishing is a necessary and sufficient condition for the existence of a transversal logical diagonal gate at the next level.
Within this framework, known algebraic conditions such as divisibility and triorthogonality are reinterpreted as necessary conditions for the existence of transversal logical diagonal gates with uniform rotation angles.
Our results identify homological obstructions governing transversal implementability and provide a conceptual foundation for a formal theory of transversal structures in quantum error correction.
\end{abstract}

\maketitle

\section{Introduction}

Transversal gates~\cite{Zeng:2007uho,Bravyi:2012rnv} play a central role in fault-tolerant quantum computation, as they prevent the propagation of physical errors within a code block. 
At the same time, their power is fundamentally limited: the Eastin--Knill theorem~\cite{Eastin:2009tem} rules out universal sets of transversal logical gates for local-error-detecting quantum codes. 
Understanding which logical gates can nevertheless be realized transversally and under what structural conditions remains a central problem in quantum error correction.

A particularly important class of fault-tolerant operations is given by logical diagonal gates implemented by transversal Pauli $Z$ rotations.
Algebraic conditions on the parity-check matrices governing the implementability of such gates have been extensively studied, ranging from weight-based criteria such as divisibility and triorthogonality~\cite{Ward:1981divisible,Rengaswamy:2020fyi,Hu:2021snt,Haah:2018uxe,Bravyi:2012lxn,Shi:2024hqn,Jain:2024zdq}, through parity constraints across the Clifford hierarchy~\cite{Campbell:2012olh,Anderson:2014jvy,Anderson:2014voa,Jochym-oconnor:2017odu,Camps-moreno:2026usd}, to structural generalizations for self-dual codes and XP stabilizer formalisms~\cite{Tansuwannont:2025riy,Tansuwannont:2026bmw,Webster:2023cpy,Webster:2022kdn}.
However, whether these criteria share a common algebraic origin, and whether the absence of certain transversal gates can be detected beyond the uniform-angle assumption, have remained unclear.

In this work, we develop a homological framework for analyzing transversal Pauli $Z$ rotations in general quantum Calderbank--Shor--Steane (CSS) codes~\cite{Calderbank:1995dw,Steane:1996ghp,Terhal:2013vbm}, formulated in terms of the chain complex~\cite{Bombin:2006cd,Kitaev:1997quantum} over the coefficient ring $\mathbb{Z}_{2^m}$.
Our results reveal two conceptually distinct layers of structure.
The first is a \emph{fixed-level classification}: at each level $m$, the set of transversal logical diagonal gates admits a homological description governed by the first cohomology group and a $\mathbb{Z}_{2^m}$-module defined with a Hadamard-product structure.
The second is an \emph{inter-level lifting problem}: whether a gate at level $m$ can be refined to level $m+1$ is governed by obstruction maps arising from level-dependent logicality conditions.
Our main results are summarized in the following theorems.

\begin{theorem}[Homological classification]
\label{thm:intro-classification}
At each level $m \geq 1$, the set of transversal logical diagonal gates modulo logical identities is classified by 
\begin{align}
V_{LD}^{(m)} \cong H^1(C;\mathbb{Z}_2) \otimes_H S_m,
\end{align}
where $H^1(C;\mathbb{Z}_2)$ is the first cohomology group of the CSS chain complex, $S_m$ is a $\mathbb{Z}_{2^m}$-module, and $\otimes_H$ denotes the Hadamard product between modules.
\end{theorem}

\begin{theorem}[Obstruction to transversal implementability]
\label{thm:intro-lifting}
A logical $Z$ operator admits a transversal implementation of its $\pi/2^m$ rotation if and only if there exists a sequence of vectors $\{[\theta^{(\nu)}] \in V_{LD}^{(\nu)}\}_{\nu=1}^m$ such that
\begin{align}
  \beta_1^{(\nu)}([\theta^{(\nu)}]) = 0, \quad \beta_2^{(\nu)}([\theta^{(\nu)}]) = 0,
\end{align}
where each $\beta_i^{(\nu)}$ is a map from $V_{LD}^{(\nu)}$ to a quotient module $T_i^{(\nu)}$ over $\mathbb{Z}_{2^\nu}$.
\end{theorem}

\Cref{thm:intro-classification} (corresponding to \Cref{thm:transversal-classification} below) extends the standard homological description of logical Pauli operators to logical diagonal gates of higher Clifford levels.
\Cref{thm:intro-lifting} (corresponding to \Cref{cor:transversal_implementability} below) shows that refinement to finer angles is controlled by the simultaneous vanishing of the obstruction maps $\beta_1^{(\nu)}$ and $\beta_2^{(\nu)}$, which capture compatibility and coefficient-extension obstructions, respectively.

We further clarify that these obstruction maps admit a homological interpretation. 
If we formally consider only coefficient extension from $\mathbb{Z}_{2^m}$ to $\mathbb{Z}_{2^{m+1}}$, the first obstruction automatically vanishes, and the second obstruction agrees with the Bockstein homomorphism associated with the short exact sequence
\begin{align}
0 \to \mathbb{Z}_2 \to \mathbb{Z}_{2^{m+1}} \to \mathbb{Z}_{2^m} \to 0.
\end{align}
This suggests that the refinement problem can be viewed, at least in part, as an obstruction problem for lifting homology classes under the coefficient extension. 

Within this framework, previously known algebraic conditions such as divisibility~\cite{Ward:1981divisible,Rengaswamy:2020fyi,Hu:2021snt} and triorthogonality~\cite{Bravyi:2012lxn} can be reinterpreted as necessary conditions for the existence of transversal logical diagonal gates with uniform rotation angles.
These arise as part of the full logicality condition for the uniform rotation-angle case, and do not characterize the vanishing of the obstruction maps.
This provides a unified perspective on transversal implementability, clarifying that logicality and lifting are governed by distinct algebraic structures.

The remainder of this paper is organized as follows.
In \Cref{sec:preliminaries}, we review CSS codes and their chain-complex formulation, introducing the homological description of logical operators. 
In \Cref{sec:classification}, we develop a homological classification of transversal logical diagonal gates at a fixed level $m$, expressed in terms of homological data and Hadamard-product structures. 
In \Cref{sec:lifting_problem}, we formulate the lifting problem and derive necessary and sufficient conditions for its solvability in terms of obstruction maps. 
In \Cref{sec:discussion}, we develop a chain-complex formulation of the lifting problem, illustrate the framework with the Steane code, and reinterpret known algebraic criteria such as divisibility and triorthogonality within our framework.
\Cref{sec:conclusion} concludes the paper with a summary and outlook.

\section{Preliminaries}
\label{sec:preliminaries}

\subsection{Quantum CSS codes}

We consider a quantum CSS code~\cite{Calderbank:1995dw,Steane:1996ghp} on $n$ physical qubits specified by two binary parity-check matrices
\begin{align}
H_X \in \mathbb{Z}_2^{m_X \times n}, 
\quad
H_Z \in \mathbb{Z}_2^{m_Z \times n}.
\end{align}
The associated stabilizer generators are given by
\begin{align}
S_X^{(i)} \coloneqq \prod_{j=1}^{n} X_j^{(H_X)_{ij}},
\quad
S_Z^{(k)} \coloneqq \prod_{j=1}^{n} Z_j^{(H_Z)_{kj}},
\end{align}
and the commutativity condition reads
\begin{align}
\label{eq:css_commutativity}
H_X H_Z^T = 0 \pmod 2 .
\end{align}

The codespace $\mathcal{H}_c$ is defined as the common $+1$ eigenspace of all stabilizers:
\begin{align}
\mathcal{H}_c \coloneqq \{\ket{\psi} \mid \forall i,\forall k, S_X^{(i)} \ket{\psi} = S_Z^{(k)} \ket{\psi} = \ket{\psi}\}.
\end{align}
A convenient basis of code states is obtained by projecting computational basis states $\ket{\gamma}$ with $\gamma \in \ker H_Z$ onto the $+1$ eigenspace of $X$-type stabilizers:
\begin{align}
  \label{eq:logical_basis_states}
\ket{[\gamma]} 
\coloneqq \prod_i \frac{1+ S_X^{(i)}}{2} \ket{\gamma}
= \frac{1}{2^{m_X}} \sum_{h_x \in \im H_X^T} \ket{\gamma + h_x}.
\end{align}
Hence logical basis states are labeled by the quotient space
\begin{align}
  \label{eq:logical_basis_labeling}
H^1(C;\mathbb{Z}_2) \coloneqq \ker H_Z / \im H_X^T,
\end{align}
which agrees with the first cohomology group of the CSS chain complex, which will be introduced in \Cref{sec:chain_complex}.

\subsection{Logical operators}

A unitary operator $U$ is a logical operator if it preserves the codespace,
\begin{align}
U \ket{\psi} \in \mathcal{H}_c
\quad \forall \ket{\psi} \in \mathcal{H}_c.
\end{align}
Equivalently, for any stabilizer $S$, the group commutator
\begin{align}
[U,S] \coloneqq U^\dagger S^\dagger U S
\end{align}
acts trivially on the codespace $\mathcal{H}_c$.

An important subclass of logical operators is the logical identity operators, which are defined as the logical operators acting as the identity on the codespace:
\begin{align}
  \label{eq:def_of_logical_identity_operator}
U \ket{\psi} = \ket{\psi}
\quad \forall \ket{\psi} \in \mathcal{H}_c.
\end{align}
These operators include not only stabilizers, but also Clifford or non-Clifford operators such as the controlled-stabilizer gates.
It has been shown that commutativity with all logical Pauli operators as well as stabilizers on the codespace $\mathcal{H}_c$ gives a necessary and sufficient condition for logical operators to be logical identities~\cite{Haruna:2025piy}, which will be used in the following sections.

A basic class of logical operators is given by transversal Pauli operators. 
For $\gamma \in \mathbb{Z}_2^n$, define the transversal $Z$ operator as
\begin{align}
\overline{Z}(\gamma) \coloneqq \prod_{i=1}^{n} Z_i^{\gamma_i}.
\end{align}
By studying the action of this operator on the logical basis states (\Cref{eq:logical_basis_states}), we find that it preserves the codespace if and only if
\begin{align}
H_X \gamma = 0 \pmod 2,
\end{align}
and thus logical $Z$ operators are specified by $\gamma \in \ker H_X$.
Furthermore, $\overline{Z}(\gamma)$ acts trivially if $\gamma \in \im H_Z^T$.
Hence their logical action is classified by taking the quotient of these two subspaces, yielding
\begin{align}
  H_1(C;\mathbb{Z}_2) \coloneqq \ker H_X / \im H_Z^T,
\end{align}
which agrees with the first homology group of the CSS chain complex, described in the next subsection.

Similarly, transversal $X$ operators, defined as
\begin{align}
  \overline{X}(\gamma) \coloneqq \prod_{i=1}^{n} X_i^{\gamma_i}, 
\end{align}
are logical if $\gamma \in \ker H_Z$, and logically trivial if $\gamma \in \im H_X^T$.
Therefore, their logical action is classified by
\begin{align}
  H^1(C;\mathbb{Z}_2) = \ker H_Z / \im H_X^T,
\end{align}
which coincides with the labeling space of logical basis states~\Cref{eq:logical_basis_labeling}.

These quotient structures provide a natural algebraic description of logical operators as equivalence classes modulo logical identity operators, forming the basis for the homological formulation used in this work.

\subsection{Chain-complex formulation}
\label{sec:chain_complex}

The structure of CSS codes can be reformulated as a length-2 chain complex over $\mathbb{Z}_2$:
\begin{align}
C \colon C_2 \xrightarrow{\;\partial_2\;} C_1 \xrightarrow{\;\partial_1\;} C_0,
\end{align}
where
\begin{align}
C_2 \cong \mathbb{Z}_2^{m_Z},\quad
C_1 \cong \mathbb{Z}_2^{n},\quad
C_0 \cong \mathbb{Z}_2^{m_X},
\end{align}
and
\begin{align}
\partial_2 \coloneqq H_Z^T,
\quad
\partial_1 \coloneqq H_X .
\end{align}
The commutativity condition~\Cref{eq:css_commutativity} is equivalent to
\begin{align}
\partial_1 \partial_2 = 0,
\end{align}
so that $C\coloneqq (C_\bullet,\partial)$ forms a chain complex.

Within this framework, logical operators admit a homological interpretation.
Logical $Z$ operators correspond to 1-cycles,
\begin{align}
Z_1(C;\mathbb{Z}_2) \coloneqq \ker \partial_1 = \ker H_X,
\end{align}
while $Z$-type stabilizers correspond to 1-boundaries,
\begin{align}
B_1(C;\mathbb{Z}_2) \coloneqq \im \partial_2 = \im H_Z^T.
\end{align}
Therefore, logical $Z$ operators are classified by the first homology group
\begin{align}
H_1(C;\mathbb{Z}_2)
\coloneqq Z_1 / B_1
= \ker H_X / \im H_Z^T.
\end{align}

Dually, logical $X$ operators are classified by the first cohomology group of the dual (cochain) complex, obtained by reversing the arrows and transposing the boundary maps:
\begin{align}
C^0 \xrightarrow{\;\partial_1^T\;} C^1 \xrightarrow{\;\partial_2^T\;} C^2,
\end{align}
where $\partial_1^T = H_X^T$ and $\partial_2^T = H_Z$, with $C^k = C_{2-k}$ as $\mathbb{Z}_2$-modules.
The first cohomology group is then
\begin{align}
H^1(C;\mathbb{Z}_2) \coloneqq \ker \partial_2^T / \im \partial_1^T = \ker H_Z / \im H_X^T.
\end{align}

This homological formulation provides a unified description of logical operators and serves as the foundation for the classification of transversal logical diagonal gates developed in the subsequent sections.

\section{Homological classification of transversal logical diagonal gates}
\label{sec:classification}

In this section, we consider transversal Pauli $Z$ rotations with discrete angles $\pi/2^{m-1}$ and characterize the resulting logical diagonal gates in terms of homological data of the CSS code.

\subsection{Transversal diagonal gates and logical action}

For $\theta \in \mathbb{Z}_{2^m}^n$, define the transversal diagonal operator
\begin{align}
U(\theta/2^{m-1})
\coloneqq
\prod_{j=1}^n 
\exp\left(
  i\pi \frac{\theta_j}{2^{m-1}} \frac{I-Z_j}{2}
\right).
\end{align}
Here $(I-Z_j)/2$ is the projector onto the $\ket{1}$ eigenstate of $Z_j$, so the operator applies a phase $e^{i\pi\theta_j/2^{m-1}}$ to each qubit in state $\ket{1}$ and the identity to state $\ket{0}$.
The integer $\theta_j$ is understood modulo $2^m$, reflecting the $2\pi$-periodicity of the phase.

This operator admits the equivalent expression
\begin{align}
  U(\theta/2^{m-1}) = \prod_{j=1}^n (Z_j^{1/2^{m-1}})^{\theta_j},
\end{align}
and thus belongs to the $m$-th level of the Clifford hierarchy~\cite{Gottesman:1999tea}, because conjugation by $U(\theta/2^{m-1})$ maps any Pauli operator to an operator at level $m-1$.
For $m=1,2$, $U(\theta/2^{m-1})$ reduces to a product of $Z$ or $S$ operators, while it includes non-Clifford gates such as $T$ and $\sqrt{T}$ for $m \ge 3$.

It acts diagonally on computational basis states $\ket{\gamma}$ as
\begin{align}
U(\theta/2^{m-1}) \ket{\gamma}
=
\exp\left(
  \frac{i\pi}{2^{m-1}} \inp{\theta}{\gamma}
\right)
\ket{\gamma},
\end{align}
where
\begin{align}
\inp{\theta}{\gamma} \coloneqq \sum_{j=1}^n \theta_j \gamma_j
\end{align}
is the standard inner product.
Throughout this work, this inner product is computed in $\mathbb{Z}$, then reduced modulo $2^m$ when needed.

We now characterize the logical action of $U(\theta/2^{m-1})$.
Following the general framework of logical operators reviewed in \Cref{sec:preliminaries}, we distinguish between operators that preserve the codespace and those that act trivially on it.

For this purpose, we define the set of rotation angles $\theta$ that give logical transversal diagonal operators by
\begin{align}
V_L^{(m)}
\coloneqq
\left\{
\theta \in \mathbb{Z}_{2^m}^n
\,\middle|\,
U(\theta/2^{m-1}) \text{ preserves } \mathcal{H}_c
\right\},
\end{align}
and that of logically trivial ones by
\begin{multline}
V_{LI}^{(m)}
\coloneqq
\\
\left\{
\theta \in \mathbb{Z}_{2^m}^n
\,\middle|\,
U(\theta/2^{m-1}) \text{ acts as identity on } \mathcal{H}_c
\right\}.
\end{multline}
Note that these sets cannot be vector spaces, because $\mathbb{Z}_{2^m}$ is not a field (but a ring) for $m \geq 2$, and should be regarded as submodules of $\mathbb{Z}_{2^m}^n$.

Clearly, $V_{LI}^{(m)}$ is a submodule of $V_L^{(m)}$ because any logical identity operator is a logical operator.
Therefore, we can define the quotient space $V_L^{(m)} / V_{LI}^{(m)}$ to classify the logical action of transversal diagonal gates modulo trivial ones as
\begin{align}
V_{LD}^{(m)} \coloneqq V_L^{(m)} / V_{LI}^{(m)}.
\end{align}
This quotient is the higher-Clifford-level analog of the homological classification of logical Pauli operators, where logical operators are identified modulo logical identities.
In the following, we determine $V_L^{(m)}$ and $V_{LI}^{(m)}$ in terms of the CSS chain complex and derive a homological description of $V_{LD}^{(m)}$.

\subsection{Characterization of logical and trivial transversal diagonal operators}

Let us derive explicit conditions on $\theta \in \mathbb{Z}_{2^m}^n$ under which the transversal operator $U(\theta/2^{m-1})$ is logical or logically trivial.

Because $U(\theta/2^{m-1})$ is diagonal, it commutes with all $Z$-type stabilizers.
Thus, the nontrivial constraint for logicality arises from commutation with $X$-type stabilizers.
For $\gamma \in \mathbb{Z}_2^n$, its commutator with a transversal $X$ operator satisfies
\begin{align}
  \label{eq:commutator_with_X}
[U, \overline{X}(\gamma)]
=
\exp\left(
\frac{i\pi}{2^{m-1}} \inp{\theta}{\gamma}
\right)
U\left(-\frac{2\,\theta \circ \gamma}{2^{m-1}}\right),
\end{align}
where $\circ$ denotes the componentwise (Hadamard) product.

The commutator with the $X$-type stabilizers $\overline{X}(\gamma)$ ($\gamma \in \im H_X^T$) must act trivially on $\mathcal{H}_c$ for $U(\theta/2^{m-1})$ to be a logical operator.
In addition, the commutator with logical $X$ operators $\overline{X}(\gamma)$ ($\gamma \in \ker H_Z$) must also act trivially for $U(\theta/2^{m-1})$ to be a logical identity operator, as we mentioned in \Cref{eq:def_of_logical_identity_operator}.

From the above commutation relation, we see that the logical triviality of the commutator is equivalent to vanishing inner products $\inp{\theta}{\gamma}$ modulo $2^m$ and the logical triviality of $U(\theta\circ \gamma/2^{m-2})$.
The point is that the logicality condition at level $m$ implies a logical triviality condition at level $m-1$ with the Hadamard product $\theta \circ \gamma$, and this recursive structure continues down to level 1.

Applying the commutator relation successively with vectors $k_z \in \ker H_Z$, we obtain a nested sequence of constraints indexed by level $\nu = 1, \ldots, m$.
Concretely, for any $h_x \in \im H_X^T$ and any $k_z^{(1)},\dots,k_z^{(m-1)} \in \ker H_Z$, we must have for $1 \le \nu \le m$ that
\begin{align}
\label{eq:logical_operator_condition}
\inp{\theta}{
2^{\nu-1}
\, h_x \circ k_z^{(1)} \circ \cdots \circ k_z^{(\nu-1)}
}
= 0
\end{align}
modulo $2^m$.
For $\nu=1$, we define $h_x \circ k_z^{(1)} \circ \cdots \circ k_z^{(\nu-1)} = h_x$.
This shows that $\theta$ must be orthogonal modulo $2^m$ to vectors generated by iterated Hadamard products of elements of $\im H_X^T$ and $\ker H_Z$ with appropriate powers of two.

To express this compactly, define the Hadamard product $\otimes_H$ of two submodules $V,W \subset \mathbb{Z}_{2^\nu}^n$ as follows:
\begin{definition}[Hadamard product of submodules]
  Let $V,W$ be submodules of $\mathbb{Z}_{2^\nu}^n$ for some integer $\nu \ge 1$.
  We define their Hadamard product $V \otimes_H W$ as the submodule of $\mathbb{Z}_{2^\nu}^n$ generated by the componentwise products of vectors in $V$ and $W$:
\begin{align}
V \otimes_H W \coloneqq \spn_{\mathbb{Z}_{2^\nu}}\{ v \circ w \mid v \in V,\; w \in W \},
\end{align}
and its iterates 
\begin{align}
V^{\otimes_H \nu}
\coloneqq
\overbrace{V \otimes_H \cdots \otimes_H V}^{\nu},
\quad
V^{\otimes_H 0} \coloneqq \{1_n\},
\end{align}
where $1_n$ is the all-ones vector,
\begin{align}
  1_n \coloneqq \begin{pmatrix}1 & \cdots & 1\end{pmatrix}^T.
\end{align}
\end{definition}

\begin{remark}
The Hadamard product is different from the usual tensor product and gives a submodule of $\mathbb{Z}_{2^\nu}^n$.
\end{remark}

Using this definition, the logicality condition~\Cref{eq:logical_operator_condition} can be expressed as
\begin{align}
  \label{eq:logical_operator_condition_orthogonality}
\inp{\theta}{a} = 0 \pmod{2^m},
\end{align}
where $a$ is an element of the submodule constructed by taking the Hadamard product of $\im H_X^T$ with iterated Hadamard products of $\ker H_Z$ scaled by powers of two:
\begin{align}
a \in 2^{\nu-1} \im H_X^T \otimes_H (\ker H_Z)^{\otimes_H (\nu-1)}.
\end{align}

Furthermore, because this must hold for all $1 \le \nu \le m$, we can combine these constraints into a single condition by summing over $\nu$.
Formally, define
\begin{align}
S_m \coloneqq \bigoplus_{\nu=1}^{m} 2^{\nu-1} (\ker H_Z)^{\otimes_H (\nu-1)},
\end{align}
where each $(\ker H_Z)^{\otimes_H (\nu-1)} \subset \mathbb{Z}_2^n$ is embedded into $\mathbb{Z}_{2^m}^n$ via $\mathbb{Z}_2 \hookrightarrow \mathbb{Z}_{2^m}$, and the direct sum is taken over submodules of $\mathbb{Z}_{2^m}^n$.
Thus, $S_m$ is a submodule of $\mathbb{Z}_{2^m}^n$ (not a vector space over $\mathbb{Z}_2$, because different terms carry different powers of two).
Intuitively, $S_m$ collects all level-$\nu$ constraint vectors scaled by $2^{\nu-1}$ to reflect the 2-adic structure of the logicality conditions.

Using this $S_m$, the above condition~\eqref{eq:logical_operator_condition_orthogonality} is equivalent to
\begin{align}
\theta \in (\im H_X^T \otimes_H S_m)^\perp,
\end{align}
and hence the space of logical transversal diagonal operators is given by the orthogonal complement as
\begin{align}
\label{eq:V_L}
V_L^{(m)} = (\im H_X^T \otimes_H S_m)^\perp.
\end{align}

Next, we characterize the logical identity operators.
Repeating the same argument with $\im H_X^T$ replaced by $\ker H_Z$, we find that $U(\theta/2^{m-1})$ acts trivially on the codespace if and only if
\begin{align}
\label{eq:logical_identity_condition}
\inp{\theta}{
2^{\nu-1}
\, k_z^{(1)} \circ \cdots \circ k_z^{(\nu)}
}
= 0 \pmod{2^m}
\end{align}
for any $ 1 \le \nu \le m$ and all $k_z^{(a)} \in \ker H_Z$.
This yields the characterization of logically trivial transversal diagonal operators as
\begin{align}
\label{eq:V_LI}
V_{LI}^{(m)} = (\ker H_Z \otimes_H S_m)^\perp.
\end{align}

Therefore, both logical and logically trivial transversal diagonal operators admit a unified description in terms of orthogonality to Hadamard-product submodules derived from the CSS chain complex.
This type of characterization of these sets has been previously discussed in the literature, such as in \cite{Bravyi:2012lxn,Anderson:2014voa,Camps-moreno:2026usd}.
However, it has not been fully clarified how transversal logical diagonal gates are classified based on the quotient of these modules, and how this classification relates to the homological structure of the code, which we address in the next subsection.

\subsection{Classification of transversal logical diagonal gates}

We now combine the characterizations of $V_L^{(m)}$ and $V_{LI}^{(m)}$ to obtain the classification of transversal logical diagonal gates, taking the quotient of logical operators by trivial ones.
Using \Cref{eq:V_L} and \Cref{eq:V_LI}, this quotient $V_{LD}^{(m)}$ can be written as
\begin{align}
V_{LD}^{(m)}
=
(\im H_X^T \otimes_H S_m)^\perp
\Big/
(\ker H_Z \otimes_H S_m)^\perp.
\end{align}

To relate this expression to the homological structure of the CSS code, we use \Cref{lem:dual_quotient}, which states that for submodules $A,B$ of $\mathbb{Z}_{2^m}^n$ with $B \subset A$, the quotient of their orthogonal complements satisfies
\begin{align}
A / B \cong B^\perp / A^\perp.
\end{align}
The proof of this lemma is given in \Cref{app:duality_quotient_submodule}, and relies on the properties of the inner product and the structure of submodules over $\mathbb{Z}_{2^m}$.
Applying \Cref{lem:dual_quotient} to the current setting, we have
\begin{align}
V_{LD}^{(m)}
\cong
(\ker H_Z \otimes_H S_m)
\Big/
(\im H_X^T \otimes_H S_m).
\end{align}

Here, we introduce the Hadamard product between a quotient module and an ordinary one as follows:
\begin{definition}[Hadamard product between quotient and submodule]
  Let $A,B,C$ be submodules of $\mathbb{Z}_{2^\nu}^n$ such that $B \subset A$ with some integer $\nu$.
  We define the Hadamard product between elements $[a]$ and $c$ of the quotient $A/B$ and the submodule $C$ as follows:
  \begin{align}
    [a] \circ c \coloneqq [a \circ c]_{B \otimes_H C},
  \end{align}
  where $[a]$ is the equivalence class of $a \in A$ in the quotient $A/B$, and $[a \circ c]_{B \otimes_H C}$ is the equivalence class of $a \circ c$ in the quotient $(A \otimes_H C) / (B \otimes_H C)$.

  Using this definition, we also define the Hadamard product between the quotient $A/B$ and the submodule $C$ as the set of all such products:
  \begin{align}
    (A/B) \otimes_H C \coloneqq \spn_{\mathbb{Z}_{2^\nu}}\{ [a] \circ c \mid a \in A,\; c \in C \}.
  \end{align}
\end{definition}

\begin{remark}
  This Hadamard-product module gives a submodule of $\mathbb{Z}_{2^\nu}^n$ and is well defined, i.e. independent of the choice of representatives of the quotient, which is shown in \Cref{prop:hadamard_product_quotient_submodule} in \Cref{app:hadamard_product_vector_space}.
\end{remark}

Using this definition, $V_{LD}^{(m)}$ admits a natural factorization with respect to the cohomology group $H^1(C;\mathbb{Z}_2)$, as follows.
\begin{theorem}[Classification of transversal logical diagonal gates]
\label{thm:transversal-classification}
  The set $V_{LD}^{(m)}$ of transversal logical diagonal gates at level $m$ is classified as
  \begin{align}
  \label{eq:transversal_classification}
  V_{LD}^{(m)} \cong H^1(C;\mathbb{Z}_2) \otimes_H S_m,
\end{align}
where
\begin{align}
  H^1(C;\mathbb{Z}_2) & \coloneqq \ker H_Z / \im H_X^T, \\
  S_m & \coloneqq \bigoplus_{\nu=1}^{m} 2^{\nu-1} (\ker H_Z)^{\otimes_H (\nu-1)}, \\
  (\ker H_Z)^{\otimes_H (\nu-1)} & \coloneqq \overbrace{\ker H_Z \otimes_H \cdots \otimes_H \ker H_Z}^{\nu-1},
\end{align}
and $\otimes_H$ denotes the Hadamard product between modules.
\end{theorem}

\begin{proof}
  By the definition of the Hadamard product between a quotient and a submodule, we can write
  \begin{multline}
    (\ker H_Z \otimes_H S_m) \Big/ (\im H_X^T \otimes_H S_m)
    \\ \eqqcolon (\ker H_Z / \im H_X^T) \otimes_H S_m
      =  H^1(C;\mathbb{Z}_2) \otimes_H S_m
  \end{multline}
  where we have used the definition of $H^1(C;\mathbb{Z}_2)$.
\end{proof}

\begin{remark}
The identification~\Cref{eq:transversal_classification} is a classification result: it gives an isomorphism of $\mathbb{Z}_{2^m}$-modules that characterizes all possible logical diagonal actions at level $m$.
For $m=1$, this recovers
\begin{align}
V_{LD}^{(1)} = \ker H_X / \im H_Z^T = H_1(C;\mathbb{Z}_2) \cong H^1(C;\mathbb{Z}_2),
\end{align}
where we used the duality between the homology and cohomology groups of the length-$2$ chain complex reflecting the duality between logical $Z$ and $X$ operators.
Thus our result correctly reproduces the standard classification of logical $Z$ operators.
For $m \ge 2$, the contributions from $S_m$ capture higher-level phase structures arising from iterated Hadamard products of $\ker H_Z$.
\end{remark}

\begin{remark}
\label{rem:comparison_webster}
The classification obtained above is closely related to the algorithmic characterization of transversal diagonal logical operators in Ref.~\cite{Webster:2023cpy} within the XP stabilizer formalism~\cite{Webster:2022kdn}.
For level-$m$ transversal phase gates (precision $N = 2^m$ in their terminology), the defining conditions of the two frameworks are equivalent: their logical identity conditions and the conditions~\eqref{eq:logical_identity_condition} reduce to each other upon expanding sums modulo $2$ over the integers via inclusion--exclusion\footnote{Explicitly, for $x_i \in \{0,1\}$,
\begin{align}
\bigoplus_{i=1}^{t} x_i = \sum_{\emptyset \neq S \subseteq \{1,\dots,t\}} (-2)^{|S|-1} \prod_{i \in S} x_i,
\end{align}
where $\oplus$ denotes addition modulo $2$;
this identity is standard in the analysis of transversal diagonal gates, see e.g.~Eqs.~(7)--(8) of Ref.~\cite{Bravyi:2012lxn}.}, while their logical operator test characterizes the same solution set $V_L^{(m)}$ as the logicality condition~\eqref{eq:logical_operator_condition}.
Consequently, the two frameworks yield the same classification $V_{LD}^{(m)} = V_L^{(m)}/V_{LI}^{(m)}$ at each fixed level.

The two works differ in what they extract from these conditions: Ref.~\cite{Webster:2023cpy} characterizes $V_{LI}^{(m)}$ and $V_L^{(m)}$ implicitly as kernels and computes their generating sets algorithmically for each given code, whereas \Cref{thm:transversal-classification} solves the defining conditions in closed form for an arbitrary CSS code, identifying the linear-algebraic structure of the classification homologically.

Their scopes are also complementary: Ref.~\cite{Webster:2023cpy} treats logical operators built from multiqubit controlled-phase gates and extends to non-CSS stabilizer codes via embedding techniques, while the homological formulation of the present work underlies the inter-level lifting structure with obstruction maps, which will be developed in \Cref{sec:lifting_problem}.
\end{remark}

\subsection{Realization of transversal logical diagonal gates}
Although the above classification characterizes the logical action of transversal diagonal gates, it does not directly address the question of which specific $\theta \in \mathbb{Z}_{2^m}^n$ implement a given logical action $[\theta] \in V_{LD}^{(m)}$.
The isomorphism~\eqref{eq:transversal_classification} is at the level of equivalence classes and does not itself select a canonical representative.

For this purpose, it is convenient to introduce a single matrix $\widetilde{H}^{(m)}$, which organizes the logicality constraints compactly, constructed as follows.
Let $K$ be a generator matrix of $\ker H_Z$ over $\mathbb{Z}_2$, namely
\begin{align}
\ker H_Z = \spn_{\mathbb{Z}_2}\{K_1,\dots,K_r\},
\end{align}
where $K_a$ denotes the $a$-th row of $K$.
Starting from $\widetilde{H}^{(1)} \coloneqq H_X$, higher-level constraints are generated by taking componentwise products of the rows of $\widetilde{H}^{(1)}$ with those of $K$ with appropriate scaling of a power of two, and iterating this procedure.

To express this operation at the matrix level, we use the Khatri--Rao product.
For matrices
\begin{align}
A =
\begin{pmatrix}
a_1 \\ \vdots \\ a_p
\end{pmatrix},
\quad
B =
\begin{pmatrix}
b_1 \\ \vdots \\ b_q
\end{pmatrix},
\end{align}
with row vectors $a_i,b_j$, we define their row-wise Khatri--Rao product by
\begin{align}
A \odot B
\coloneqq
\begin{pmatrix}
a_1 \circ b_1 \\
a_1 \circ b_2 \\
\vdots \\
a_1 \circ b_q \\
a_2 \circ b_1 \\
\vdots \\
a_p \circ b_q
\end{pmatrix}.
\end{align}
Thus, $A \odot B$ is the matrix whose rows are all possible componentwise products of a row of $A$ with a row of $B$.
Note that this is a row-wise variant of the Khatri--Rao product; it differs from the standard column-wise definition in that rows, rather than columns, are enumerated.
In the present setting, $\widetilde{H}^{(m)} \odot K$ therefore enumerates all one-step extensions of the constraints encoded in $\widetilde{H}^{(m)}$ by one additional factor from $\ker H_Z$.

With this notation, we define $\widetilde{H}^{(m)}$ recursively by
\begin{empheq}[left=\empheqlbrace]{align}
\widetilde{H}^{(1)} &\coloneqq H_X, \\
\widetilde{H}^{(m+1)}
&\coloneqq
\begin{pmatrix}
\widetilde{H}^{(m)} \\
2\,(\widetilde{H}^{(m)} \odot K)
\end{pmatrix}.
\label{eq:Htilde_recursive_definition}
\end{empheq}
By construction, the rows of $\widetilde{H}^{(m)}$ are precisely of the form
\begin{align}
2^{\nu-1} \, h_x \circ k_z^{(1)} \circ \cdots \circ k_z^{(\nu-1)}
\quad (1 \le \nu \le m),
\end{align}
with $h_x$ running over the rows of $H_X$ and each $k_z^{(a)}$ running over the rows of $K$.
Therefore, the logicality condition~\Cref{eq:logical_operator_condition} is equivalent to the single matrix congruence
\begin{align}
\widetilde{H}^{(m)} \theta = 0 \pmod{2^m}.
\label{eq:logical_condition_Htilde}
\end{align}

In this way, $\widetilde{H}^{(m)}$ provides a compact encoding of all logicality constraints up to level $m$.
It is generally highly redundant, because many rows produced by the recursive construction coincide or are linearly dependent modulo powers of two, but this redundancy is harmless for our purposes.

To classify them up to logical identities, we need to identify the submodule $V_{LI}^{(m)}$ of logically trivial operators.
The logical identity condition~\Cref{eq:logical_identity_condition} can also be expressed in terms of a Khatri--Rao construction, by replacing $H_X$ with the generator matrix $K$ of $\ker H_Z$.
We define $\widetilde{K}^{(m)}$ recursively by
\begin{empheq}[left=\empheqlbrace]{align}
\widetilde{K}^{(1)} &\coloneqq K, \\
\widetilde{K}^{(m+1)}
&\coloneqq
\begin{pmatrix}\widetilde{K}^{(m)} \\
2\,(\widetilde{K}^{(m)} \odot K)
\end{pmatrix}.
\end{empheq}
Then the logical identity condition is equivalent to
\begin{align}
  \label{eq:logical_identity_condition_Ktilde}
\widetilde{K}^{(m)} \theta = 0 \pmod{2^m}.
\end{align}

Using these matrices, we can obtain a concrete description of the quotient $V_{LD}^{(m)} = V_L^{(m)} / V_{LI}^{(m)}$.
One subtlety must be kept in mind: because $\mathbb{Z}_{2^m}$ is not a field for $m \geq 2$, the submodules $V_L^{(m)}$ and $V_{LI}^{(m)}$ of $\mathbb{Z}_{2^m}^n$ need not be free, and hence do not admit a basis in the strict sense.
Nevertheless, any submodule $V \subset \mathbb{Z}_{2^m}^n$ is a finite abelian $2$-group, so the structure theorem of finite abelian groups yields a decomposition into cyclic submodules,
\begin{align}
  \label{eq:cyclic_decomposition}
V = \bigoplus_{i=1}^{d} \mathbb{Z}_{2^m} e_i,
\qquad
\mathbb{Z}_{2^m} e_i \cong \mathbb{Z}_{2^{m_i}},
\end{align}
where each generator $e_i$ has order $2^{m_i}$ with $1 \leq m_i \leq m$.
We call such a set $\{e_i\}$ an \emph{adapted generating set} of $V$; every element of $V$ is then written uniquely as $\sum_i \lambda_i e_i$ with $\lambda_i \in \mathbb{Z}_{2^{m_i}}$.
For kernels such as \Cref{eq:logical_condition_Htilde,eq:logical_identity_condition_Ktilde}, an adapted generating set and the orders of its generators are computable via the Smith normal form of an integer lift of the defining matrix, or equivalently via the Howell normal form over $\mathbb{Z}_{2^m}$~\cite{Webster:2023cpy}.

As a first step, we solve \Cref{eq:logical_condition_Htilde} for $\theta \in \mathbb{Z}_{2^m}^n$, and obtain a general solution
\begin{align}
\theta = \sum_{i=1}^{d^{(m)}_{L}} \lambda_i e_i,
\end{align}
where $\{e_i\}$ is an adapted generating set of $V_L^{(m)}$, $d^{(m)}_{L}$ denotes the number of its generators, and $\lambda_i \in \mathbb{Z}_{2^{m_i}}$ with $2^{m_i}$ the order of $e_i$.

Next, we identify the submodule $V_{LI}^{(m)}$ of logically trivial operators as the solution to \Cref{eq:logical_identity_condition_Ktilde} as
\begin{align}
  \label{eq:logical_identity_solution}
\theta = \sum_{j=1}^{d^{(m)}_{LI}} \alpha_j g_j,
\end{align}
where $\{g_j\}$ is an adapted generating set of $V_{LI}^{(m)}$ and $d^{(m)}_{LI}$ denotes the number of its generators.
Here, we let each $\alpha_j$ range over all of $\mathbb{Z}_{2^m}$; this parametrization of $V_{LI}^{(m)}$ is redundant but surjective, which suffices for all subsequent arguments.

Finally, the quotient $V_{LD}^{(m)} = V_L^{(m)}/V_{LI}^{(m)}$ is again a finite $\mathbb{Z}_{2^m}$-module generated by the images $[e_i]$ of the generators of $V_L^{(m)}$.
An adapted generating set of the quotient, together with the orders of its generators, follows from a further Smith-normal-form computation applied to the relation matrix expressing $\{g_j\}$ and $\{2^{m_i} e_i\}$ in terms of $\{e_i\}$.
In particular, the number of distinct logical actions realized by level-$m$ transversal diagonal gates is of the order $|V_{LD}^{(m)}|$, given by the product of the orders of these generators.

Each generator of the cyclic decomposition of $V_{LD}^{(m)}$ provides an independent logical transversal diagonal gate, and all transversal operators $U(\theta'/2^{m-1})$ equivalent to $U(\theta/2^{m-1})$ on the codespace are characterized by the vector
\begin{align}
  \theta' = \theta + \widetilde{G}^{(m)} \alpha,
\end{align}
with
\begin{align}
  \widetilde{G}^{(m)} \coloneqq
\begin{pmatrix}
 g_1 & g_2 & \cdots & g_{d^{(m)}_{LI}}
\end{pmatrix}
,
\end{align}
where $g_j$ and $\alpha$ are given in \Cref{eq:logical_identity_solution}.

In the next section, we consider refinement of the rotation angle based on this classification, namely whether a logical diagonal gate at level $m$ admits a transversal realization at level $m+1$.

\section{Lifting problem and obstruction maps}
\label{sec:lifting_problem}

\subsection{Lifting problem}

In the previous section, we obtained a homological classification of transversal logical diagonal gates at level $m$ in terms of the first cohomology group $H^1(C;\mathbb{Z}_2)$ and the auxiliary space $S_m$.
Thus, at each fixed level $m$, the logical diagonal gates by transversal rotations are described by homology with coefficients in $\mathbb{Z}_{2^m}$.

However, this classification does not answer a more subtle question: given a logical diagonal gate implementable at angle $\pi/2^{m-1}$, can it always be refined to the finer angle $\pi/2^{m}$?

This question is not answered by the fixed-level classification, and the answer is in general no.
Even if $U(\theta/2^{m-1})$ is a valid logical operator at level $m$, the na\"ive candidate $U(\theta/2^m)$ need not be logical at level $m+1$, because the logicality condition~\Cref{eq:logical_operator_condition} at level $m+1$ imposes strictly stronger constraints on $\theta$ than those at level $m$.
Therefore, refining the rotation angle requires finding a possibly different representative $\hat{\theta}$ of the same logical class, satisfying the additional constraints at level $m+1$.

More generally, we formalize this question as the following lifting problem.
\begin{definition}[Lifting problem]
Let $\theta \in V_L^{(m)} \subset \mathbb{Z}_{2^m}^n$ give a transversal logical diagonal gate $U(\theta/2^{m-1})$.
We ask whether there exists $\hat{\theta} \in \mathbb{Z}_{2^{m+1}}^n$ that satisfies
\begin{enumerate}[label=(\roman*)]
\item $U(\hat{\theta}/2^{m-1})$ implements the same logical operator as $U(\theta/2^{m-1})$, and
\item $U(\hat{\theta}/2^m)$ is a valid logical operator.
\end{enumerate}
In particular, for a given $\theta \in V_L^{(m)} \subset \mathbb{Z}_{2^m}^n$, we call this the ``level-$m$'' lifting problem for $\theta$, and we call any solution $\hat{\theta}$ to it a ``lift'' of $\theta$ from level $m$ to $m+1$ throughout this paper.
\end{definition}

The first condition ensures that $U(\hat{\theta}/2^m)$ gives a square root of $U(\theta/2^{m-1})$ on the codespace, while the second condition ensures that $U(\hat{\theta}/2^{m})$ is itself logical.
Note that $\hat{\theta}$ is not necessarily the same as $\theta$ as an element of $\mathbb{Z}_{2^{m+1}}^n$, because all we require for logical square roots is the equality of logical actions at level $m$.

The lifting problem asks whether the logical action represented by $\theta$ at level $m$ admits a representative at level $m+1$ whose transversal rotation angle is finer by a factor of two.
As we show below, the obstruction to such a lift can be captured by certain maps defined on the space of logical operators at level $m$, which we call obstruction maps.

\subsection{Constraints of lifts and obstruction maps}

We now analyze the lifting problem introduced in the previous subsection more concretely.
As derived in \Cref{sec:classification} (see \Cref{eq:logical_operator_condition}), the logicality condition for $U(\theta/2^m)$ at level $m$ requires $\theta$ to lie in the kernel of the matrix $\widetilde{H}^{(m)}$ modulo $2^m$.

For a lift $\hat{\theta} \in \mathbb{Z}_{2^{m+1}}^n$, the two lifting conditions (i) and (ii) can be reformulated as follows.
First, the requirement that $\hat{\theta}$ induces the same logical operator as $\theta$ at level $m$ is equivalent to
\begin{align}
\label{eq:first_lift_condition}
\hat{\theta}
=
\theta + \widetilde{G}^{(m)} \alpha + 2^m \omega
\pmod{2^{m+1}},
\end{align}
for some $\alpha \in \mathbb{Z}_{2^m}^{d^{(m)}_{LI}}$ and $\omega \in \mathbb{Z}_2^n$.
Second, the requirement that $U(\hat{\theta}/2^m)$ is logical at level $m+1$ is equivalent to
\begin{align}
\label{eq:second_lift_condition}
\widetilde{H}^{(m+1)} \hat{\theta} = 0 \pmod{2^{m+1}}.
\end{align}

Using the recursive structure (\Cref{eq:Htilde_recursive_definition}) of $\widetilde{H}^{(m)}$ and $\widetilde{H}^{(m+1)}$, and substituting~\Cref{eq:first_lift_condition} into~\Cref{eq:second_lift_condition}, we obtain two independent conditions as follows:
\begin{empheq}[left=\empheqlbrace]{align}
\label{eq:lift_condition_split_1}
& (\widetilde{H}^{(m)} \odot K) \bigl( \theta + \widetilde{G}^{(m)} \alpha \bigr) = 0 \pmod{2^m},
\\
\label{eq:lift_condition_split_2}
& \widetilde{H}^{(m)} \bigl( \theta + \widetilde{G}^{(m)} \alpha \bigr) = 2^m \widetilde{H}^{(m)} \omega \pmod{2^{m+1}}.
\end{empheq}
The first equation is a consistency condition that must be satisfied already at level $m$, while the second equation determines whether the lift extends to level $m+1$.
These two equations give rise to two independent obstruction maps: the first one $\beta_1^{(m)}$ measuring compatibility at the current level, and the second one $\beta_2^{(m)}$ measuring the obstruction to coefficient extension from $\mathbb{Z}_{2^m}$ to $\mathbb{Z}_{2^{m+1}}$.

We first analyze~\Cref{eq:lift_condition_split_1}.
Rewriting it, we obtain
\begin{align}
  \label{eq:first_obstruction_equation}
(\widetilde{H}^{(m)} \odot K)\theta
\in
\im\left(
(\widetilde{H}^{(m)} \odot K) \widetilde{G}^{(m)}
\right).
\end{align}
Based on this observation, we can define the following \emph{first obstruction} $\beta_1^{(m)}$ that measures whether $(\widetilde{H}^{(m)} \odot K)\theta$ can be cancelled by shifting $\theta$ without changing its logical class.

\begin{definition}[First obstruction map]
\label{def:beta1}
Define the first obstruction map $\beta_1^{(m)} : V_{LD}^{(m)} \to T_1^{(m)}$ by
\begin{align}
\beta_1^{(m)}([\theta])
\coloneqq
\left[
(\widetilde{H}^{(m)} \odot K)\theta
\right]
\in T_1^{(m)},
\end{align}
where $T_1^{(m)}$ is a quotient module defined as $\mathbb{Z}_{2^m}^{(n-m_z)r_m}$ divided by the right-hand side of \Cref{eq:first_obstruction_equation}, with $r_m$ the number of rows of $\widetilde{H}^{(m)}$.
(The precise definition of $T_1^{(m)}$ is given in \Cref{app:definition_of_obstruction_maps}.)
\end{definition}

\begin{remark}
$\beta_1^{(m)}$ is well defined as a map on the quotient $V_{LD}^{(m)}$, that is, independent of the choice of representative $\theta$ in $V_{LD}^{(m)}$, as verified in \Cref{app:definition_of_obstruction_maps}.
\end{remark}

\begin{remark}
By construction, $\beta_1^{(m)}([\theta])=0$ is a necessary condition for the solvability of~\Cref{eq:lift_condition_split_1}.
Hence, this first obstruction map captures the first hurdle in the lifting problem: whether the logical operator at level $m$ can be adjusted within its logical class to satisfy the additional constraints at level $m+1$.
\end{remark}

Next, we analyze the second equation, \Cref{eq:lift_condition_split_2}, with the assumption that $\beta_1^{(m)}([\theta])=0$.
Then there exists $\alpha_0$ such that
\begin{align}
(\widetilde{H}^{(m)} \odot K)
\bigl(
\theta + \widetilde{G}^{(m)} \alpha_0
\bigr)
= 0
\pmod{2^m}.
\end{align}
The general solution is $\alpha=\alpha_0+\eta$, where
\begin{align}
\eta \in \ker\left((\widetilde{H}^{(m)} \odot K) \widetilde{G}^{(m)}\right).
\end{align}
Because $\theta, \widetilde{G}^{(m)}\alpha_0 \in V_L^{(m)}$ for any $\alpha_0$, they satisfy $\widetilde{H}^{(m)}\theta = 0 \pmod{2^m}$ and $\widetilde{H}^{(m)}\widetilde{G}^{(m)} = 0 \pmod{2^m}$.
Then, we find that all terms are divisible by $2^m$ in the second equation~\Cref{eq:lift_condition_split_2}.
Dividing it by $2^m$ modulo $2$, we obtain
\begin{multline}
\label{eq:second_obstruction_equation}
\frac{1}{2^m}
\widetilde{H}^{(m)}
\bigl(
\theta + \widetilde{G}^{(m)} \alpha_0
\bigr)
\\ = \frac{1}{2^m} \widetilde{H}^{(m)} \widetilde{G}^{(m)} \eta + \widetilde{H}^{(m)} \omega \pmod{2}.
\end{multline}

The right-hand side lies in
\begin{align}
\im\left(
\eval{\frac{\widetilde{H}^{(m)} \widetilde{G}^{(m)}}{2^m}}
_{\ker((\widetilde{H}^{(m)} \odot K) \widetilde{G}^{(m)})}
\right)
+
\im \widetilde{H}^{(m)}.
\end{align}
The \emph{second obstruction} $\beta_2^{(m)}$ measures the failure of the $2^m$-divisible remainder $\widetilde{H}^{(m)}(\theta+\widetilde{G}^{(m)}\alpha_0)/2^m$ to lie in the image of $\widetilde{H}^{(m)}$ with some correction term.

\begin{definition}[Second obstruction map]
\label{def:beta2}
Suppose $\beta_1^{(m)}([\theta])=0$, and let $\alpha_0$ be any solution to~\Cref{eq:lift_condition_split_1}.
Define the second obstruction map $\beta_2^{(m)} : \ker \beta_1^{(m)} \to T_2^{(m)}$ by
\begin{align}
\beta_2^{(m)}([\theta])
\coloneqq
\left[
\frac{1}{2^m}
\widetilde{H}^{(m)}
\bigl(
\theta + \widetilde{G}^{(m)} \alpha_0
\bigr)
\right]
\in T_2^{(m)},
\end{align}
where $T_2^{(m)}$ is a quotient module defined as $\mathbb{Z}^{r_m}_2$ divided by the right-hand side of \Cref{eq:second_obstruction_equation} and $r_m$ is the number of rows of $\widetilde{H}^{(m)}$.
(The precise definition of $T_2^{(m)}$ is given in \Cref{app:definition_of_obstruction_maps}.)
\end{definition}

\begin{remark}
  The well definedness of $\beta_2^{(m)}$ (independence of the choice of $\alpha_0$ and of the representative $\theta$) is verified in \Cref{app:definition_of_obstruction_maps}.
\end{remark}

\begin{remark}
  By construction, $\beta_2^{(m)}([\theta])=0$ is necessary for the solvability of~\Cref{eq:lift_condition_split_2}.
  Hence, it measures the failure of $\theta$ to satisfy the extension condition~\Cref{eq:lift_condition_split_2} after solving~\Cref{eq:lift_condition_split_1}.
\end{remark}

In the next subsection, we prove that the vanishing of both $\beta_1^{(m)}([\theta])$ and $\beta_2^{(m)}([\theta])$ is also sufficient for the existence of a lift, thus giving a complete characterization of the existence of a solution to the lifting problem.

\subsection{Lifting criterion and transversal implementability of finer rotations}

Combining the constraints from the two obstruction maps derived in the previous subsection, we obtain a necessary and sufficient condition for the existence of lifts, and thus for the existence of transversal implementations of logical diagonal gates with finer rotation angles.

\begin{theorem}[Obstructions to existence of lifts]
\label{thm:lifting-criterion}
A lift of $\theta \in V_L^{(m)}$ from level $m$ to level $m+1$ exists if and only if both obstruction maps vanish:
\begin{align}
\beta_1^{(m)}([\theta]) = 0,
\quad
\beta_2^{(m)}([\theta]) = 0.
\end{align}
\end{theorem}

\begin{proof}
The necessity follows directly from the construction of the obstruction maps in the previous subsection.

For sufficiency, suppose that both obstruction maps vanish.
Then there exist $\alpha_0$ and $\omega$ satisfying the equations
\begin{align}
  \label{eq:def_alpha0}
  (\widetilde{H}^{(m)} \odot K) (\theta + \widetilde{G}^{(m)} \alpha_0) = 0 \pmod{2^m},
\end{align}
  and 
\begin{multline}
  \label{eq:def_eta_omega}
  \widetilde{H}^{(m)} (\theta + \widetilde{G}^{(m)} (\alpha_0 + \eta)+ 2^m \omega) \\
    = 0 \pmod{2^{m+1}},
\end{multline}
  with $\eta$ satisfying
\begin{align}
  \label{eq:def_eta}
  (\widetilde{H}^{(m)} \odot K) \widetilde{G}^{(m)} \eta & = 0 \pmod{2^m}.
\end{align}
Defining $\hat{\theta}$ as 
\begin{align}
  \hat{\theta}
  \coloneqq
  \theta + \widetilde{G}^{(m)} (\alpha_0 + \eta) + 2^m \omega, \pmod{2^{m+1}}
\end{align}
we can verify that $\hat{\theta}$ satisfies the two lifting conditions.

\textit{Step 1 (compatibility).}
First, we check that $\hat{\theta}$ induces the same logical operator as $\theta$ at level $m$.
This $\hat{\theta}$ is calculated modulo $2^m$ as
\begin{align}
\hat{\theta} = \theta + \widetilde{G}^{(m)} (\alpha_0 + \eta) \pmod{2^m}.
\end{align}
Because $\alpha_0 + \eta \eqqcolon \alpha$ is an element of $\mathbb{Z}_{2^m}^{d_{LI}^{(m)}}$, we can rewrite $\hat{\theta}$ as
\begin{align}
\hat{\theta} = \theta + \widetilde{G}^{(m)} \alpha \pmod{2^m},
\end{align}
so the first lifting condition is satisfied.

\textit{Step 2 (extension).}
Next, we check that $U(\hat{\theta}/2^m)$ is a valid logical operator.
$\widetilde{H}^{(m+1)} \hat{\theta}$ is calculated modulo $2^{m+1}$ as
\begin{align}
\widetilde{H}^{(m+1)} \hat{\theta}
&=
\begin{pmatrix}
\widetilde{H}^{(m)} (\theta + \widetilde{G}^{(m)} (\alpha_0 + \eta) + 2^m \omega) \\
2\, (\widetilde{H}^{(m)} \odot K) (\theta + \widetilde{G}^{(m)} (\alpha_0 + \eta))
\end{pmatrix}
\\ & = 0 \pmod{2^{m+1}}.
\end{align}
In the second line, the first row vanishes by~\Cref{eq:def_eta_omega}.
For the second row, all the terms vanish owing to~\Cref{eq:def_alpha0} and \Cref{eq:def_eta}.

Hence, $\hat{\theta}$ satisfies both lifting conditions and gives a valid lift.
\end{proof}

Iterating \Cref{thm:lifting-criterion} yields a necessary and sufficient condition for the existence of multi-level lifts.

\begin{corollary}[Multilevel lifting]
\label{cor:multi_level_lifting}
A lift of $\theta \in V_L^{(1)} = \ker H_X$ from level $1$ to $m$ exists if and only if there exists a sequence $\{\theta^{(\nu)}\}_{\nu=1}^{m}$ such that
\begin{align}
\theta^{(1)} = \theta,
\quad
\theta^{(\nu)} \in V_L^{(\nu)},
\end{align}
and for every $\nu = 1,\dots,m-1$,
\begin{align}
\beta_1^{(\nu)}([\theta^{(\nu)}]) = 0,
\quad
\beta_2^{(\nu)}([\theta^{(\nu)}]) = 0.
\end{align}
\end{corollary}

\begin{proof}
The sufficiency follows by recursively applying \Cref{thm:lifting-criterion}.

For necessity, let $\hat{\theta}$ be a level-$m$ lift of $\theta$, and define
\begin{align}
\theta^{(\nu)} \coloneqq \hat{\theta} \pmod{2^\nu}.
\end{align}
Because $\hat{\theta} \in V_L^{(m)}$, it satisfies the logicality conditions at all intermediate levels: for each $\nu \le m$, the constraints of $\widetilde{H}^{(\nu)}$ appear as a subblock of $\widetilde{H}^{(m)}$, so $\widetilde{H}^{(\nu)} \hat{\theta} = 0 \pmod{2^m}$.
Reducing modulo $2^\nu$ gives $\theta^{(\nu)} = \hat{\theta} \pmod{2^\nu}$, and therefore
\begin{align}
\widetilde{H}^{(\nu)} \theta^{(\nu)} = \widetilde{H}^{(\nu)} \hat{\theta} = 0 \pmod{2^\nu},
\end{align}
yielding $\theta^{(\nu)} \in V_L^{(\nu)}$.
Each $\theta^{(\nu+1)}$ therefore qualifies as a lift of $\theta^{(\nu)}$ from level $\nu$ to $\nu+1$, implying the vanishing of the obstruction maps.
\end{proof}

Although the above results characterize the existence of lifts, we must additionally control the pairing with a basis of the logical states to realize a specific logical Pauli $Z$ rotation.
Let $\{[\gamma_b]\}$ be a basis of $H^1(C;\mathbb{Z}_2)$ and $\gamma_b$ be a representative of $[\gamma_b]$ in $\ker H_Z$.
Then, let us assume that $\theta \in V^{(1)}_L$ and its lift $\hat{\theta} \in V^{(m)}_L$ to level $m$ satisfy
\begin{align}
  \inp{\theta}{\gamma_b} &= \delta_{ab} \pmod{2}, \\
  \inp{\hat{\theta}}{\gamma_b} &= \delta_{ab} \pmod{2^m}.
\end{align}
The transversal rotation $U(\hat{\theta}/2^{m-1})$ acts on the logical state $\ket{[\gamma_b]}$ as
\begin{align}
  U(\hat{\theta}/2^{m-1})\ket{[\gamma_b]} = \exp\left( \frac{i\pi}{2^{m-1}} \delta_{ab}\right) \ket{[\gamma_b]},
\end{align}
which agrees with the $\pi/2^{m-1}$ rotation of the logical Pauli $Z$ operator $\overline{Z}(\theta)$.

The following proposition shows that this pairing can always be preserved under lifting.

\begin{proposition}[Inner-product preservation]
\label{prop:inner_product_preservation}
Let $\{[\gamma_i]\}$ be a basis of $H^1(C;\mathbb{Z}_2)$ and suppose that $\theta \in V_L^{(m)}$ satisfies
\begin{align}
\inp{\theta}{\gamma_i} = c_i \pmod{2^m}
\end{align}
with some $c_i \in \mathbb{Z}_{2^m}$ for all $i$.
If $\theta$ admits a lift to level $m+1$, then there exists another lift $\hat{\theta}$ such that
\begin{align}
\inp{\hat{\theta}}{\gamma_i} = c_i \pmod{2^{m+1}}
\end{align}
for all $i$.
\end{proposition}

\begin{proof}
Let $\hat{\theta}$ be any lift of $\theta$, which can be written as
\begin{align}
\hat{\theta} = \theta + \widetilde{G}^{(m)} \alpha + 2^m \omega
\end{align}
with some $\alpha \in \mathbb{Z}_{2^m}^{d_{LI}^{(m)}}$ and $\omega \in \mathbb{Z}_2^n$.
Then the pairing of $\hat{\theta}$ with $\gamma_i$ is calculated as
\begin{multline}
\inp{\hat{\theta}}{\gamma_i} = \inp{\theta}{\gamma_i} + \inp{\widetilde{G}^{(m)} \alpha}{\gamma_i}
\\ + \inp{2^m \omega}{\gamma_i}
\pmod{2^{m+1}}.
\end{multline}

The second term can be written as
\begin{align}
\inp{\widetilde{G}^{(m)} \alpha}{\gamma_i}
= \sum_{j=1}^n (\gamma_i)_j (\widetilde{G}^{(m)} \alpha)_j.
\end{align}
The divisibility by $2^m$ for all $i$ follows from the assumption $\widetilde{G}^{(m)} \alpha \in V_{LI}^{(m)}$ and the first block of the logical identity condition~\Cref{eq:logical_identity_condition_Ktilde}.
So we can write the pairing modulo $2^{m+1}$ as
\begin{align}
\inp{\hat{\theta}}{\gamma_i}
=
c_i + 2^m a_i \pmod{2^{m+1}}
\end{align}
for some $a_i \in \mathbb{Z}_2$.

Then, we can define a new lift $\hat{\theta}'$ by shifting $\hat{\theta}$ as
\begin{align}
\hat{\theta}'
\coloneqq
\hat{\theta} + 2^m \sum_i a_i \tilde{\gamma}_i,
\end{align}
where $\{[\tilde{\gamma}_i]\}$ is a basis of $H_1(C;\mathbb{Z}_2)$ with representatives $\tilde{\gamma}_i \in \ker H_X$ satisfying
\begin{align}
\inp{\tilde{\gamma}_i}{\gamma_j} = \delta_{ij} \pmod{2}.
\end{align}
The existence of such a basis is guaranteed by the duality between the logical $Z$ and $X$ operators, reflecting the nondegeneracy of the pairing between $H_1(C;\mathbb{Z}_2)$ and $H^1(C;\mathbb{Z}_2)$.

Because $\tilde{\gamma}_i \in \ker H_X$, it can be shown that 
\begin{align}
\widetilde{H}^{(m)}(2^m \tilde{\gamma}_i)=0 \pmod{2^{m+1}}
\end{align}
by the recursive structure of $\widetilde{H}^{(m)}$.
So $\hat{\theta}'$ satisfies the logicality condition and remains a valid lift.

In addition, the pairing of $\hat{\theta}'$ with $\gamma_i$ is calculated as
\begin{align}
\inp{\hat{\theta}'}{\gamma_i} &= c_i + 2\cdot 2^m a_i 
\\ &= c_i \pmod{2^{m+1}}
\end{align}
for all $i$.
Thus, we have constructed a lift $\hat{\theta}'$ that preserves the pairing with the basis $\{[\gamma_i]\}$, as claimed.

\end{proof}

Combining these results, we obtain a necessary and sufficient condition for the transversal implementability of a $\pi/2^{m-1}$ rotation of a given logical $Z$ operator.

\begin{corollary}[Transversal implementability]
\label{cor:transversal_implementability}
Let $\overline{Z}(\gamma)$ be a logical $Z$ operator represented with $\gamma \in \ker H_X$.
A transversal implementation of its $\pi/2^{m}$ rotation exists if and only if there exists a sequence $\{\theta^{(\nu)}\}_{\nu=1}^{{m}}\, (\theta^{(\nu)} \in V_L^{(\nu)}, \theta^{(1)}=\gamma)$ such that
\begin{align}
\beta_1^{(\nu)}([\theta^{(\nu)}]) = 0,
\quad
\beta_2^{(\nu)}([\theta^{(\nu)}]) = 0
\end{align}
for all $\nu$.
\end{corollary}

\begin{proof}
  This follows from \Cref{cor:multi_level_lifting} and \Cref{prop:inner_product_preservation}.
\end{proof}

Finally, we find that given a CSS code and one of its logical $Z$ operators, the existence of a transversal $\pi/2^m$ rotation reduces to the existence of a sequence of lifts $\{\theta^{(\nu)}\}_{\nu=1}^m$.
This is equivalent to the existence of $\alpha, \omega$, and $\eta$ satisfying the two lift conditions \Cref{eq:lift_condition_split_1} and \Cref{eq:lift_condition_split_2}.
Because these equations are linear equations over $\mathbb{Z}_{2^\nu}$, we can solve them by standard linear algebraic methods such as Gaussian elimination or Smith normal form, and any solution yields vanishing of the obstruction maps.

The lifting problem admits a natural homological interpretation.
Recall that $\widetilde{G}^{(m)}$ is a matrix whose columns are generators of $V_{LI}^{(m)}$, and $\widetilde{H}^{(m)}$ is a matrix whose kernel gives the set of logical operators at level $m$.
Using these matrices, the classification $V_{LD}^{(m)}$ of transversal logical diagonal gates at level $m$ can be expressed as
\begin{align}
\label{eq:logical_diagonal_gates_as_quotient}
V_{LD}^{(m)} = \ker \widetilde{H}^{(m)} \big/ \im \widetilde{G}^{(m)}.
\end{align}
Moreover, the matrices $\widetilde{G}^{(m)}$ and $\widetilde{H}^{(m)}$ satisfy 
\begin{align}
\widetilde{H}^{(m)} \widetilde{G}^{(m)} = 0 \pmod{2^m},
\end{align}
which follows directly from the definitions.

These observations allow us to define a chain complex over $\mathbb{Z}_{2^m}$:
\begin{align}
\mathcal{C}^{(m)} \colon
\mathcal{C}_2^{(m)}
\xrightarrow{\;\widetilde{G}^{(m)}\;}
\mathcal{C}_1^{(m)}
\xrightarrow{\;\widetilde{H}^{(m)}\;}
\mathcal{C}_0^{(m)},
\end{align}
where
\begin{align}
\mathcal{C}_2^{(m)} \coloneqq \mathbb{Z}_{2^m}^{d^{(m)}_{LI}},
\quad
\mathcal{C}_1^{(m)} \coloneqq \mathbb{Z}_{2^m}^{n},
\quad
\mathcal{C}_0^{(m)} \coloneqq \mathbb{Z}_{2^m}^{r_m},
\end{align}
and $r_m$ denotes the number of rows of $\widetilde{H}^{(m)}$.

The first homology group of this chain complex is
\begin{align}
H_1(\mathcal{C}^{(m)};\mathbb{Z}_{2^m}) \coloneqq \ker \widetilde{H}^{(m)} \big/ \im \widetilde{G}^{(m)},
\end{align}
and therefore agrees with the classification of transversal logical diagonal gates at level $m$:
\begin{align}
H_1(\mathcal{C}^{(m)};\mathbb{Z}_{2^m}) = V_{LD}^{(m)}.
\end{align}

For $m=1$, this reduces to the standard CSS chain complex with $\widetilde{H}^{(1)} = H_X$ and $\widetilde{G}^{(1)} = H_Z^T$, yielding
\begin{align}
H_1(\mathcal{C}^{(1)};\mathbb{Z}_2) = \ker H_X / \im H_Z^T,
\end{align}
which classifies logical Pauli $Z$ operators. Thus, the complexes $\mathcal{C}^{(m)}$ provide a natural extension of the CSS chain complex to higher-level transversal diagonal gates.

We now interpret the lifting problem in this framework. Given $\theta \in V_L^{(m)}$, its logical action corresponds to a homology class
\begin{align}
[\theta] \in H_1(\mathcal{C}^{(m)};\mathbb{Z}_{2^m}).
\end{align}
The lifting problem asks whether this class admits a lift to a class
\begin{align}
[\hat{\theta}] \in H_1(\mathcal{C}^{(m+1)};\mathbb{Z}_{2^{m+1}})
\end{align}
under consistency at level $m$.
From this viewpoint, the obstruction maps introduced in \Cref{def:beta1,def:beta2} measure the failure of such a lift to exist.

\section{Discussion}
\label{sec:discussion}

In \Cref{sec:classification}, we established a homological classification of transversal logical diagonal gates at each fixed level $m$, expressed as $V_{LD}^{(m)} \cong H^1(C;\mathbb{Z}_2) \otimes_H S_m$.
In \Cref{sec:lifting_problem}, we formulated the refinement to finer angles as a lifting problem and showed that its solvability is completely characterized by the vanishing of two obstruction maps $\beta_1^{(m)}$ and $\beta_2^{(m)}$.
Together, these results separate transversal implementability into a fixed-level homological structure and an inter-level obstruction-theoretic structure.

In this section, we develop further structural interpretations of this framework, the relationship between the chain-complex formulation of the lifting problem and the Bockstein homomorphisms, the Steane code as a concrete example of the lifting problem, and a reinterpretation of known algebraic conditions such as divisibility and triorthogonality within our framework.

\subsection{Chain complex and the lifting problem at higher levels}

We reinterpret the classification of transversal logical diagonal gates and the lifting problem in a homological framework.

As shown in the previous section, the classification of transversal logical diagonal gates at level $m$ can be expressed in terms of the homology group of a chain complex defined by the matrices $\widetilde{G}^{(m)}$ and $\widetilde{H}^{(m)}$.
Then, the lifting problem can be regarded as a question of whether a homology class in $H_1(\mathcal{C}^{(m)})$ admits a lift to a homology class in $H_1(\mathcal{C}^{(m+1)})$.

In this context, the obstruction maps $\beta_1^{(m)}$ and $\beta_2^{(m)}$ can be interpreted as measuring the failure of a homology class to admit such a lift, which can be clarified as follows.
Let us consider the case where we formally fix the boundary operators between levels $m$ and $m+1$,
\begin{align}
\widetilde{H}^{(m+1)} \to \widetilde{H}^{(m)}, \quad \widetilde{G}^{(m+1)} \to \widetilde{G}^{(m)},
\end{align}
and only the coefficient extension from $\mathbb{Z}_{2^m}$ to $\mathbb{Z}_{2^{m+1}}$ is considered.
Then, the boundary operators should satisfy the compatibility condition at level $m+1$,
\begin{align}
  \label{eq:compatibility_condition_formal}
  \widetilde{H}^{(m)} \widetilde{G}^{(m)} = 0 \pmod{2^{m+1}}.
\end{align}
In particular, from the decomposition of $\widetilde{H}^{(m+1)}$ in \Cref{eq:Htilde_recursive_definition}, the second block of $\widetilde{H}^{(m+1)}$ is formally set to zero,
\begin{align}
  \label{eq:formal_setting_Htilde_odot_K}
  \widetilde{H}^{(m)} \odot K \to 0.
\end{align}

In this setting, the obstruction maps become more transparent.
The first obstruction map $\beta_1^{(m)}$ always vanishes,
\begin{align}
\beta_1^{(m)}([\theta]) = [0\cdot \theta] = 0,
\end{align}
from the above formal condition (\Cref{eq:formal_setting_Htilde_odot_K}) on $\widetilde{H}^{(m)} \odot K$, so the first hurdle in the lifting problem is automatically cleared.
Furthermore, the second obstruction map $\beta_2^{(m)}$ reduces to a simpler form, given by
\begin{align}
\beta_2^{(m)}([\theta]) = \left[ \frac{1}{2^m} \widetilde{H}^{(m)} \theta \right] 
\in \mathbb{Z}_2^{r_m} \big/ \im \widetilde{H}^{(m)},
\end{align}
which is precisely the Bockstein homomorphism associated with the short exact sequence of coefficient rings,
\begin{align}
0 \to \mathbb{Z}_2 
\xrightarrow{\times 2^m} 
\mathbb{Z}_{2^{m+1}} 
\xrightarrow{\mod 2^m} 
\mathbb{Z}_{2^m} 
\to 0.
\end{align}
Therefore, the solvability of the lifting problem reduces to the vanishing of the Bockstein homomorphism, which is a standard obstruction-theoretic condition in algebraic topology for lifting homology classes under coefficient extensions.
This homomorphism has appeared only sporadically in the context of studies of quantum error correction~\cite{Feng:2025bww,Barkeshli:2022edm,Breuckmann:2024axh,Bauer:2024qpc}, and its relationship to the transversal logical diagonal gates has not been fully recognized, to our knowledge.

However, in our setting, the boundary operators $\widetilde{H}^{(m)}$ and $\widetilde{G}^{(m)}$ themselves depend on $m$: specifically, $\widetilde{H}^{(m+1)}$ has strictly more rows than $\widetilde{H}^{(m)}$ from the recursive construction, and the chain complex $\mathcal{C}^{(m)}$ changes between levels.
A standard Bockstein homomorphism, by contrast, arises when the chain complex is \emph{fixed} and only the coefficient ring is extended; here the complex and the coefficients change simultaneously.
As a result, $\beta_2^{(m)}$ is a Bockstein-type map rather than an exact Bockstein homomorphism.
In addition, the first obstruction map $\beta_1^{(m)}$ likewise reflects the additional constraints imposed by the extra rows of $\widetilde{H}^{(m+1)}$ that are absent from $\widetilde{H}^{(m)}$.

We leave a homological-algebra formulation of the lifting problem with the level-dependent chain complexes $\mathcal{C}^{(m)}$ for future work.
One possible approach is to regard each $\mathcal{C}_i^{(m)}$ as a submodule of $\mathcal{C}_i^{(m+1)}$ and study compatibility between the embedding map and the boundary operators.
If they are compatible with each other, the whole chain complex ${\mathcal{C}_i^{(m)}}$ can form a double complex defined by filtrations of $\mathcal{C}_i^{(m)}$ by level $m$, and we may analyze the lifting problem by spectral-sequence methods.

\subsection{Examples of the lifting problem: Steane code}

We illustrate our obstruction framework by analyzing the lifting problem for the Steane code~\cite{Steane:1996ghp}, which serves as a minimal nontrivial example demonstrating that the obstruction maps $\beta_1^{(m)}$ and $\beta_2^{(m)}$ can be nontrivial.
The results are summarized as follows.

\begin{corollary}[Steane code]
\label{cor:steane}
For the Steane code, $\beta_1^{(1)}$ and $\beta_2^{(1)}$ vanish for the unique logical class, so a level-$2$ lift exists and a transversal $S$ gate is implementable.
However, $\beta_1^{(2)}$ and $\beta_2^{(2)}$ do not vanish, so no level-$3$ lift exists: no transversal $T$ gate is implementable, even with non-uniform rotation angles $\theta \in \mathbb{Z}_8^7$.
\end{corollary}

We verify this by explicit computation below.
The Steane code is a self-dual $[[7,1,3]]$ CSS code defined by the $[7,4,3]$ Hamming code.
Let
\begin{align}
H \coloneqq
\begin{pmatrix}
0 & 0 & 0 & 1 & 1 & 1 & 1 \\
0 & 1 & 1 & 0 & 0 & 1 & 1 \\
1 & 0 & 1 & 0 & 1 & 0 & 1
\end{pmatrix}
\end{align}
Then $H_X = H_Z = H$.

To consider the lifting problem, we first need to construct the matrix $\widetilde{H}^{(m)}$, which is defined recursively with $H_X$ and $K$ (a basis of $\ker H_Z$).
A basis of $\ker H_Z$ can be chosen as the union of a representative of $H^1(C;\mathbb{Z}_2)$ and a basis of $\im H_X^T$ because
\begin{align}
  \ker H_Z & \cong (\ker H_Z/\im H_X^T) \oplus \im H_X^T 
  \\ &= H^1(C;\mathbb{Z}_2) \oplus \im H_X^T.
\end{align}
We take the logical $X$ representative as $\tilde{\gamma} \coloneqq 1_7$ and the remaining basis vectors as rows of $H_X^T$, giving a basis of $\ker H_Z$ as
\begin{align}
K \coloneqq
\begin{pmatrix}
1_7^T \\
H
\end{pmatrix}.
\end{align}

\textit{Level $m=2$ (transversal $S$ gate, no obstruction).}
The matrix $\widetilde{H}^{(2)}$ reduces to
\begin{align}
\widetilde{H}^{(2)} \simeq
\begin{pmatrix}
H \\
2H_2
\end{pmatrix}, \quad H_2 \coloneqq
\begin{pmatrix}
h_1 \circ h_2 \\
h_2 \circ h_3 \\
h_3 \circ h_1
\end{pmatrix},
\end{align}
where $h_i$ denotes the $i$-th row of $H$.
Note that the equality ($\simeq$) is not exact, meaning that the right-hand side is obtained by removing the redundant rows and performing row operations, which do not change the solution space of the logicality condition.

The logicality condition
\begin{align}
  \widetilde{H}^{(2)} \theta = 0 \pmod{4}
\end{align}
is equivalent to
\begin{subequations}
  \begin{empheq}[left=\empheqlbrace]{align}
    H\theta &= 0 \pmod{4}, \\
    H_2 \theta &= 0 \pmod{2}.
  \end{empheq}
\end{subequations}
The second equation gives the modulo 2 condition and corresponds to the first obstruction map $\beta_1^{(1)}$, while the first equation gives the modulo 4 condition and corresponds to the second obstruction map $\beta_2^{(1)}$.

The vector $1_7$ satisfies all these conditions, which means that the image of the obstruction maps vanishes,
\begin{align}
\beta_1^{(1)}([1_7]) = 0, \quad \beta_2^{(1)}([1_7]) = 0,
\end{align}
and therefore defines a transversal phase gate.
However, the logical phase is determined by
\begin{align}
\inp{1_7}{1_7} = 7 = 3 \pmod{4},
\end{align}
so $U(1_7/2)$ implements a logical $S^3 = S^\dagger$ gate rather than $S$.
To obtain the correct phase, we use inner-product preservation.
Adding $2\cdot 1_7$ gives
\begin{align}
\hat{\theta} = 3 \cdot 1_7,
\end{align}
which satisfies
\begin{align}
\inp{\hat{\theta}}{1_7} = 1 \pmod{4}.
\end{align}
Thus $U(3\cdot 1_7/2) = \prod_i S_i^\dagger$ implements the logical $S$ gate.

\textit{Level $m=3$ (transversal $T$ gate, obstruction).}
The matrix $\widetilde{H}^{(3)}$ reduces to
\begin{align}
\widetilde{H}^{(3)} \simeq
\begin{pmatrix}
  H \\
  2H_2 \\
  4H_3
\end{pmatrix},
\quad
H_3 \coloneqq
\begin{pmatrix}
h_1 \circ h_2 \circ h_3
\end{pmatrix},
\end{align}
and the first and second blocks ($H$ and $2H_2$) correspond to the second obstruction map $\beta_2^{(2)}$, while the third block ($4H_3$) corresponds to the first obstruction map $\beta_1^{(2)}$.
The logicality condition 
\begin{align}
  \label{eq:steane_logicality_condition_level3}
\widetilde{H}^{(3)} \theta = 0 \pmod{8}
\end{align}
reduces to
\begin{subequations}
\begin{empheq}[left=\empheqlbrace]{align}
H\theta &= 0 \pmod{8}, \\
H_2\theta &= 0 \pmod{4}, \\
H_3\theta &= 0 \pmod{2}.
\end{empheq}
\end{subequations}

Although we can verify the nonexistence of a lift of $1_7$ to level $3$ by tedious calculation, we instead show this by directly analyzing the solutions to the logicality condition at level 3 (\Cref{eq:steane_logicality_condition_level3}).
The general solution $\theta^{(3)}$ to it is given by
\begin{align}
\theta^{(3)} = 2c_0 1_7 + 4\sum_{i=1}^3 c_i h_i,
\end{align}
where $c_0 \in \mathbb{Z}_4$ and $c_1,c_2,c_3 \in \mathbb{Z}_2$ are free parameters; the vectors $2\cdot 1_7$ and $4 h_i$ form an adapted generating set (\Cref{eq:cyclic_decomposition}) of $V_L^{(3)}$, whose generators have orders $4$ and $2$, respectively.
This illustrates that $V_L^{(m)}$ is not a free $\mathbb{Z}_{2^m}$-module in general.
Because every solution is divisible by $2$, the inner product of $\theta^{(3)}$ with $1_7$ is always even:
\begin{align}
\inp{\theta^{(3)}}{1_7} = 0 \neq 1 \pmod{2}.
\end{align}

As we discussed in \Cref{prop:inner_product_preservation}, the logical $\pi/4$ rotation of the logical $Z$ operator requires any lift $\hat{\theta}$ to preserve the inner product, i.e. $\inp{\hat{\theta}}{1_7} = 1 \pmod{8}$.
However, any solution to the logicality condition at level $3$ has an inner product with $1_7$ that is always even and cannot be unity modulo $8$.
Therefore, we reach the conclusion that no lift of $1_7$ to level $3$ exists, and hence no transversal $T$ gate exists for the Steane code, even with non-uniform angles $\theta \in \mathbb{Z}_8^7$.

This result is consistent with the Eastin--Knill theorem~\cite{Eastin:2009tem}, which implies that the Steane code does not admit a transversal $T$ gate, and our framework provides a concrete obstruction-theoretic explanation for this fact.

\subsection{Interpretation of known transversal gate conditions}

We reinterpret several known algebraic conditions for transversal phase gates within the present framework.
In particular, we consider the divisible-code condition~\cite{Ward:1981divisible,Rengaswamy:2020fyi,Hu:2021snt} and the triorthogonality condition~\cite{Bravyi:2012lxn}.

Let us focus on transversal operators with uniform rotation angle, corresponding to $\theta = 1_n$.
In this case, the logical condition reduces to a constraint on the Hamming weight of rows of $\widetilde{H}^{(m)}$
\begin{align}
  (\widetilde{H}^{(m)} 1_n)_i  = \sum_{j=1}^n \widetilde{H}^{(m)}_{ij} = \wt(\widetilde{H}^{(m)}_i).
\end{align}
Thus, a logicality of $1_n$ is equivalent to requiring that each row of $\widetilde{H}^{(m)}$ has weight divisible by $2^m$.

We first consider divisible codes.
A CSS code is called $\Delta$-divisible if every $X$-stabilizer has weight divisible by $\Delta$.
For $\Delta = 4$ and $8$, these correspond to doubly-even and triply-even codes.
It is known that $8$-divisibility is a necessary condition for transversal $T$ gates with uniform angles~\cite{Rengaswamy:2020fyi}.

In our framework, this condition appears as part of the logical condition for $\theta = 1_n$.
Indeed, because $\widetilde{H}^{(m)}$ contains $H_X$ as its first block, we have
\begin{align}
(H_X 1_n)_i = \wt((H_X)_i) = 0 \pmod{2^m}.
\end{align}
Therefore, a $2^m$-divisibility of $X$-stabilizers is necessary for a level-$m$ transversal gate with uniform angles.
However, this condition is not sufficient in general, because additional constraints arise from higher Hadamard-product blocks of $\widetilde{H}^{(m)}$.

We next consider triorthogonal codes.
A matrix $G \in \mathbb{Z}_2^{m\times n}$ is called triorthogonal if
\begin{align}
\wt(G_i \circ G_j) &= 0 \pmod{2}, \\
\wt(G_i \circ G_j \circ G_k) &= 0 \pmod{2},
\end{align}
for all distinct indices.
From such a matrix, a CSS code is constructed by taking
\begin{align}
H_X \coloneqq G_e,
\quad
H_Z \coloneqq D,
\end{align}
where $G_e$ consists of even-weight rows of $G$, and $D$ spans $\ker G_e$.

In this setting, the logical condition at level $m=3$ includes constraints coming from triple Hadamard products.
In particular, the block $4 H_X \odot H_X \odot H_X$ in $\widetilde{H}^{(3)}$ gives
\begin{align}
(H_X \odot H_X \odot H_X) 1_n = 0 \pmod{2}.
\end{align}
Expanding this condition yields
\begin{align}
\wt((G_e)_i \circ (G_e)_j \circ (G_e)_k) = 0 \pmod{2},
\end{align}
for all $i,j,k$.
This includes the three cases
\begin{align}
\wt((G_e)_i) &= 0 \quad (i=j=k), \\
\wt((G_e)_i \circ (G_e)_j) &= 0 \quad (i\neq j = k), \\
\wt((G_e)_i \circ (G_e)_j \circ (G_e)_k) &= 0 \quad (\text{all distinct}),
\end{align}
because $(G_e)_i \circ (G_e)_i = (G_e)_i$ for any $\mathbb{Z}_2$ vector $(G_e)_i$.
These relationships recover the parity conditions underlying triorthogonality, and thus triorthogonality arises as a subset of the logical conditions for $\theta = 1_n$ at level $m=3$.
However, these conditions do not exhaust all constraints coming from $\widetilde{H}^{(3)}$, and hence are not sufficient for transversal $T$ gates in general.

In summary, these known algebraic conditions for transversal gates can be understood as necessary conditions that emerge from the logicality condition for specific choices of $\theta = 1_n$.
However, the full logicality condition for general $\theta$ includes additional constraints that are not captured by these conditions alone, so they are not sufficient in general.
This highlights the importance of the full homological framework developed in this work for systematically analyzing transversal gates beyond these special cases.

\section{Conclusion and outlook}
\label{sec:conclusion}

In this work, we have established a homological framework for transversal Pauli $Z$ rotations with discrete angles in quantum CSS codes, focusing on three central structural results.

First, at each fixed level $m$, transversal logical diagonal gates modulo logical identities admit a homological classification,
\begin{align}
V_{LD}^{(m)} \coloneqq V_L^{(m)} / V_{LI}^{(m)}
\cong H^1(C;\mathbb{Z}_2) \otimes_H S_m,
\end{align}
which extends the standard description of logical Pauli operators by incorporating level-dependent phase structure through Hadamard-product extensions.

Second, we formulated the refinement of transversal rotations as a lifting problem and proved that it is characterized within this framework by two obstruction maps,
\begin{align}
\beta_1^{(m)}, \quad \beta_2^{(m)}.
\end{align}
Their simultaneous vanishing gives a necessary and sufficient condition for the existence of a lift from level $m$ to $m+1$, providing a complete criterion for the existence of finer-angle transversal implementations.

Finally, we showed that these obstructions admit a homological interpretation.
In particular, when the boundary operators are formally fixed, the first obstruction $\beta_1^{(m)}$ vanishes automatically, and the second obstruction $\beta_2^{(m)}$ reduces to the Bockstein homomorphism associated with the short exact sequence of coefficient rings,
\begin{align}
0 \to \mathbb{Z}_2 \to \mathbb{Z}_{2^{m+1}} \to \mathbb{Z}_{2^m} \to 0,
\end{align}
so that the refinement problem becomes an obstruction problem for lifting homology classes under coefficient extension.
This identifies transversal implementability of logical diagonal gates as a Bockstein-type obstruction phenomenon, generalized by the presence of level-dependent constraints.

Overall, transversal logical diagonal gates in CSS codes are governed by homological data and their liftability, rather than by ad hoc algebraic conditions, providing a unified structural understanding of their existence.
Determining the existence of transversal implementations at finer angles requires more detailed information about the code structure than ordinary CSS homology; the obstruction maps provide a systematic way to extract this information.

Several directions for future work follow naturally.
First, it would be interesting to explore when the obstruction maps vanish for specific classes of codes, and to identify code properties that guarantee the existence of transversal gates at finer angles.
To this end, understanding the filtration structure of the chain complexes $\mathcal{C}^{(m)}$ may play a key role.
In addition, extending the present framework beyond diagonal gates, for instance to transversal finer-angle rotations of logical $X$ operators or general unitary operations, or to constructions combining diagonal gates with entangling operations (such as CZ-assisted transversal implementations), may reveal a further correspondence between quantum error correction and obstruction theory.

More broadly, our results suggest that homological methods provide a systematic language for analyzing fault-tolerant logical operations in quantum CSS codes. 
The appearance of the Bockstein-type obstruction indicates that transversal gates are controlled by obstruction theory, hinting at a deeper connection between quantum error correction and algebraic topology. 
We expect that similar homological structures may play a key role in other aspects of fault tolerance, including logical gates that preserve locality and structural constraints beyond diagonal operations.

\section*{Acknowledgments}
The author thanks Kohei Yamamoto and Yutaka Hirano at the University of Osaka for fruitful discussions.

\section*{Funding}
This work was supported by the Ministry of Education, Culture, Sports, Science and Technology (MEXT) Quantum Leap Flagship Program (MEXT Q-LEAP) Grant No. JPMXS0120319794, JST COI-NEXT Grant No. JPMJPF2014, JST Moonshot R\&D Grant Nos. JPMJMS2061 and JPMJMS256E, and JST CREST Grant No. JPMJCR24I3.

\bibliographystyle{apsrev4-2}
\bibliography{ref}

\begin{thebibliography}{30}%
\makeatletter
\providecommand \@ifxundefined [1]{%
 \@ifx{#1\undefined}
}%
\providecommand \@ifnum [1]{%
 \ifnum #1\expandafter \@firstoftwo
 \else \expandafter \@secondoftwo
 \fi
}%
\providecommand \@ifx [1]{%
 \ifx #1\expandafter \@firstoftwo
 \else \expandafter \@secondoftwo
 \fi
}%
\providecommand \natexlab [1]{#1}%
\providecommand \enquote  [1]{``#1''}%
\providecommand \bibnamefont  [1]{#1}%
\providecommand \bibfnamefont [1]{#1}%
\providecommand \citenamefont [1]{#1}%
\providecommand \href@noop [0]{\@secondoftwo}%
\providecommand \href [0]{\begingroup \@sanitize@url \@href}%
\providecommand \@href[1]{\@@startlink{#1}\@@href}%
\providecommand \@@href[1]{\endgroup#1\@@endlink}%
\providecommand \@sanitize@url [0]{\catcode `\\12\catcode `\$12\catcode
  `\&12\catcode `\#12\catcode `\^12\catcode `\_12\catcode `\%12\relax}%
\providecommand \@@startlink[1]{}%
\providecommand \@@endlink[0]{}%
\providecommand \url  [0]{\begingroup\@sanitize@url \@url }%
\providecommand \@url [1]{\endgroup\@href {#1}{\urlprefix }}%
\providecommand \urlprefix  [0]{URL }%
\providecommand \Eprint [0]{\href }%
\providecommand \doibase [0]{https://doi.org/}%
\providecommand \selectlanguage [0]{\@gobble}%
\providecommand \bibinfo  [0]{\@secondoftwo}%
\providecommand \bibfield  [0]{\@secondoftwo}%
\providecommand \translation [1]{[#1]}%
\providecommand \BibitemOpen [0]{}%
\providecommand \bibitemStop [0]{}%
\providecommand \bibitemNoStop [0]{.\EOS\space}%
\providecommand \EOS [0]{\spacefactor3000\relax}%
\providecommand \BibitemShut  [1]{\csname bibitem#1\endcsname}%
\let\auto@bib@innerbib\@empty
\bibitem [{\citenamefont {Zeng}\ \emph {et~al.}(2011)\citenamefont {Zeng},
  \citenamefont {Cross},\ and\ \citenamefont {Chuang}}]{Zeng:2007uho}%
  \BibitemOpen
  \bibfield  {author} {\bibinfo {author} {\bibfnamefont {B.}~\bibnamefont
  {Zeng}}, \bibinfo {author} {\bibfnamefont {A.}~\bibnamefont {Cross}},\ and\
  \bibinfo {author} {\bibfnamefont {I.~L.}\ \bibnamefont {Chuang}},\ }\href
  {https://doi.org/10.1109/TIT.2011.2161917} {\bibfield  {journal} {\bibinfo
  {journal} {IEEE Transactions on Information Theory}\ }\textbf {\bibinfo
  {volume} {57}},\ \bibinfo {pages} {6272} (\bibinfo {year} {2011})},\ \Eprint
  {https://arxiv.org/abs/0706.1382} {arXiv:0706.1382 [quant-ph]} \BibitemShut
  {NoStop}%
\bibitem [{\citenamefont {Bravyi}\ and\ \citenamefont
  {Koenig}(2013)}]{Bravyi:2012rnv}%
  \BibitemOpen
  \bibfield  {author} {\bibinfo {author} {\bibfnamefont {S.}~\bibnamefont
  {Bravyi}}\ and\ \bibinfo {author} {\bibfnamefont {R.}~\bibnamefont
  {Koenig}},\ }\href {https://doi.org/10.1103/PhysRevLett.110.170503}
  {\bibfield  {journal} {\bibinfo  {journal} {Physical Review Letters}\
  }\textbf {\bibinfo {volume} {110}},\ \bibinfo {pages} {170503} (\bibinfo
  {year} {2013})},\ \Eprint {https://arxiv.org/abs/1206.1609} {arXiv:1206.1609
  [quant-ph]} \BibitemShut {NoStop}%
\bibitem [{\citenamefont {Eastin}\ and\ \citenamefont
  {Knill}(2009)}]{Eastin:2009tem}%
  \BibitemOpen
  \bibfield  {author} {\bibinfo {author} {\bibfnamefont {B.}~\bibnamefont
  {Eastin}}\ and\ \bibinfo {author} {\bibfnamefont {E.}~\bibnamefont {Knill}},\
  }\href {https://doi.org/10.1103/PhysRevLett.102.110502} {\bibfield  {journal}
  {\bibinfo  {journal} {Physical Review Letters}\ }\textbf {\bibinfo {volume}
  {102}},\ \bibinfo {pages} {110502} (\bibinfo {year} {2009})},\ \Eprint
  {https://arxiv.org/abs/0811.4262} {arXiv:0811.4262 [quant-ph]} \BibitemShut
  {NoStop}%
\bibitem [{\citenamefont {Ward}(1981)}]{Ward:1981divisible}%
  \BibitemOpen
  \bibfield  {author} {\bibinfo {author} {\bibfnamefont {H.~N.}\ \bibnamefont
  {Ward}},\ }\href {https://doi.org/10.1007/BF01223730} {\bibfield  {journal}
  {\bibinfo  {journal} {Archiv der Mathematik}\ }\textbf {\bibinfo {volume}
  {36}},\ \bibinfo {pages} {485} (\bibinfo {year} {1981})}\BibitemShut
  {NoStop}%
\bibitem [{\citenamefont {Rengaswamy}\ \emph {et~al.}(2020)\citenamefont
  {Rengaswamy}, \citenamefont {Calderbank}, \citenamefont {Newman},\ and\
  \citenamefont {Pfister}}]{Rengaswamy:2020fyi}%
  \BibitemOpen
  \bibfield  {author} {\bibinfo {author} {\bibfnamefont {N.}~\bibnamefont
  {Rengaswamy}}, \bibinfo {author} {\bibfnamefont {R.}~\bibnamefont
  {Calderbank}}, \bibinfo {author} {\bibfnamefont {M.}~\bibnamefont {Newman}},\
  and\ \bibinfo {author} {\bibfnamefont {H.~D.}\ \bibnamefont {Pfister}},\
  }\href {https://doi.org/10.1109/JSAIT.2020.3012914} {\bibfield  {journal}
  {\bibinfo  {journal} {IEEE Journal on Selected Areas in Information Theory}\
  }\textbf {\bibinfo {volume} {1}},\ \bibinfo {pages} {499} (\bibinfo {year}
  {2020})},\ \Eprint {https://arxiv.org/abs/1910.09333} {arXiv:1910.09333
  [quant-ph]} \BibitemShut {NoStop}%
\bibitem [{\citenamefont {Hu}\ \emph {et~al.}(2022)\citenamefont {Hu},
  \citenamefont {Liang},\ and\ \citenamefont {Calderbank}}]{Hu:2021snt}%
  \BibitemOpen
  \bibfield  {author} {\bibinfo {author} {\bibfnamefont {J.}~\bibnamefont
  {Hu}}, \bibinfo {author} {\bibfnamefont {Q.}~\bibnamefont {Liang}},\ and\
  \bibinfo {author} {\bibfnamefont {R.}~\bibnamefont {Calderbank}},\ }\href
  {https://doi.org/10.22331/q-2022-09-08-802} {\bibfield  {journal} {\bibinfo
  {journal} {Quantum}\ }\textbf {\bibinfo {volume} {6}},\ \bibinfo {pages}
  {802} (\bibinfo {year} {2022})},\ \Eprint {https://arxiv.org/abs/2109.13481}
  {arXiv:2109.13481 [quant-ph]} \BibitemShut {NoStop}%
\bibitem [{\citenamefont {Haah}(2018)}]{Haah:2018uxe}%
  \BibitemOpen
  \bibfield  {author} {\bibinfo {author} {\bibfnamefont {J.}~\bibnamefont
  {Haah}},\ }\href {https://doi.org/10.1103/PhysRevA.97.042327} {\bibfield
  {journal} {\bibinfo  {journal} {Physical Review A}\ }\textbf {\bibinfo
  {volume} {97}},\ \bibinfo {pages} {042327} (\bibinfo {year} {2018})},\
  \Eprint {https://arxiv.org/abs/1709.08658} {arXiv:1709.08658 [quant-ph]}
  \BibitemShut {NoStop}%
\bibitem [{\citenamefont {Bravyi}\ and\ \citenamefont
  {Haah}(2012)}]{Bravyi:2012lxn}%
  \BibitemOpen
  \bibfield  {author} {\bibinfo {author} {\bibfnamefont {S.}~\bibnamefont
  {Bravyi}}\ and\ \bibinfo {author} {\bibfnamefont {J.}~\bibnamefont {Haah}},\
  }\href {https://doi.org/10.1103/PhysRevA.86.052329} {\bibfield  {journal}
  {\bibinfo  {journal} {Physical Review A}\ }\textbf {\bibinfo {volume} {86}},\
  \bibinfo {pages} {052329} (\bibinfo {year} {2012})},\ \Eprint
  {https://arxiv.org/abs/1209.2426} {arXiv:1209.2426 [quant-ph]} \BibitemShut
  {NoStop}%
\bibitem [{\citenamefont {Shi}\ \emph {et~al.}(2024)\citenamefont {Shi},
  \citenamefont {Lu}, \citenamefont {Kim},\ and\ \citenamefont
  {Sol{\'e}}}]{Shi:2024hqn}%
  \BibitemOpen
  \bibfield  {author} {\bibinfo {author} {\bibfnamefont {M.}~\bibnamefont
  {Shi}}, \bibinfo {author} {\bibfnamefont {H.}~\bibnamefont {Lu}}, \bibinfo
  {author} {\bibfnamefont {J.-L.}\ \bibnamefont {Kim}},\ and\ \bibinfo {author}
  {\bibfnamefont {P.}~\bibnamefont {Sol{\'e}}},\ }\href
  {https://doi.org/10.1007/s11128-024-04485-9} {\bibfield  {journal} {\bibinfo
  {journal} {Quantum Information Processing}\ }\textbf {\bibinfo {volume}
  {23}},\ \bibinfo {pages} {280} (\bibinfo {year} {2024})},\ \Eprint
  {https://arxiv.org/abs/2408.09685} {arXiv:2408.09685 [cs.IT]} \BibitemShut
  {NoStop}%
\bibitem [{\citenamefont {Jain}\ and\ \citenamefont
  {Albert}(2025)}]{Jain:2024zdq}%
  \BibitemOpen
  \bibfield  {author} {\bibinfo {author} {\bibfnamefont {S.~P.}\ \bibnamefont
  {Jain}}\ and\ \bibinfo {author} {\bibfnamefont {V.~V.}\ \bibnamefont
  {Albert}},\ }\href {https://doi.org/10.1109/JSAIT.2025.3570832} {\bibfield
  {journal} {\bibinfo  {journal} {IEEE Journal on Selected Areas in Information
  Theory}\ }\textbf {\bibinfo {volume} {6}},\ \bibinfo {pages} {127} (\bibinfo
  {year} {2025})},\ \Eprint {https://arxiv.org/abs/2408.12752}
  {arXiv:2408.12752 [quant-ph]} \BibitemShut {NoStop}%
\bibitem [{\citenamefont {Campbell}\ \emph {et~al.}(2012)\citenamefont
  {Campbell}, \citenamefont {Anwar},\ and\ \citenamefont
  {Browne}}]{Campbell:2012olh}%
  \BibitemOpen
  \bibfield  {author} {\bibinfo {author} {\bibfnamefont {E.~T.}\ \bibnamefont
  {Campbell}}, \bibinfo {author} {\bibfnamefont {H.}~\bibnamefont {Anwar}},\
  and\ \bibinfo {author} {\bibfnamefont {D.~E.}\ \bibnamefont {Browne}},\
  }\href {https://doi.org/10.1103/PhysRevX.2.041021} {\bibfield  {journal}
  {\bibinfo  {journal} {Physical Review X}\ }\textbf {\bibinfo {volume} {2}},\
  \bibinfo {pages} {041021} (\bibinfo {year} {2012})},\ \Eprint
  {https://arxiv.org/abs/1205.3104} {arXiv:1205.3104 [quant-ph]} \BibitemShut
  {NoStop}%
\bibitem [{\citenamefont {Anderson}\ \emph {et~al.}(2014)\citenamefont
  {Anderson}, \citenamefont {{Duclos-Cianci}},\ and\ \citenamefont
  {Poulin}}]{Anderson:2014jvy}%
  \BibitemOpen
  \bibfield  {author} {\bibinfo {author} {\bibfnamefont {J.~T.}\ \bibnamefont
  {Anderson}}, \bibinfo {author} {\bibfnamefont {G.}~\bibnamefont
  {{Duclos-Cianci}}},\ and\ \bibinfo {author} {\bibfnamefont {D.}~\bibnamefont
  {Poulin}},\ }\href {https://doi.org/10.1103/PhysRevLett.113.080501}
  {\bibfield  {journal} {\bibinfo  {journal} {Physical Review Letters}\
  }\textbf {\bibinfo {volume} {113}},\ \bibinfo {pages} {080501} (\bibinfo
  {year} {2014})},\ \Eprint {https://arxiv.org/abs/1403.2734} {arXiv:1403.2734
  [quant-ph]} \BibitemShut {NoStop}%
\bibitem [{\citenamefont {Anderson}\ and\ \citenamefont
  {{Jochym-O'Connor}}(2016)}]{Anderson:2014voa}%
  \BibitemOpen
  \bibfield  {author} {\bibinfo {author} {\bibfnamefont {J.~T.}\ \bibnamefont
  {Anderson}}\ and\ \bibinfo {author} {\bibfnamefont {T.}~\bibnamefont
  {{Jochym-O'Connor}}},\ }\href {https://doi.org/10.26421/QIC16.9-10-3}
  {\bibfield  {journal} {\bibinfo  {journal} {Quantum Information and
  Computation}\ }\textbf {\bibinfo {volume} {16}},\ \bibinfo {pages} {771}
  (\bibinfo {year} {2016})},\ \Eprint {https://arxiv.org/abs/1409.8320}
  {arXiv:1409.8320 [quant-ph]} \BibitemShut {NoStop}%
\bibitem [{\citenamefont {{Jochym-O'Connor}}\ \emph {et~al.}(2018)\citenamefont
  {{Jochym-O'Connor}}, \citenamefont {Kubica},\ and\ \citenamefont
  {Yoder}}]{Jochym-oconnor:2017odu}%
  \BibitemOpen
  \bibfield  {author} {\bibinfo {author} {\bibfnamefont {T.}~\bibnamefont
  {{Jochym-O'Connor}}}, \bibinfo {author} {\bibfnamefont {A.}~\bibnamefont
  {Kubica}},\ and\ \bibinfo {author} {\bibfnamefont {T.~J.}\ \bibnamefont
  {Yoder}},\ }\href {https://doi.org/10.1103/PhysRevX.8.021047} {\bibfield
  {journal} {\bibinfo  {journal} {Physical Review X}\ }\textbf {\bibinfo
  {volume} {8}},\ \bibinfo {pages} {021047} (\bibinfo {year} {2018})},\ \Eprint
  {https://arxiv.org/abs/1710.07256} {arXiv:1710.07256 [quant-ph]} \BibitemShut
  {NoStop}%
\bibitem [{\citenamefont {{Camps-Moreno}}\ \emph {et~al.}(2026)\citenamefont
  {{Camps-Moreno}}, \citenamefont {L{\'o}pez}, \citenamefont {Matthews},
  \citenamefont {Rengaswamy},\ and\ \citenamefont
  {{San-Jos{\'e}}}}]{Camps-moreno:2026usd}%
  \BibitemOpen
  \bibfield  {author} {\bibinfo {author} {\bibfnamefont {E.}~\bibnamefont
  {{Camps-Moreno}}}, \bibinfo {author} {\bibfnamefont {H.~H.}\ \bibnamefont
  {L{\'o}pez}}, \bibinfo {author} {\bibfnamefont {G.~L.}\ \bibnamefont
  {Matthews}}, \bibinfo {author} {\bibfnamefont {N.}~\bibnamefont
  {Rengaswamy}},\ and\ \bibinfo {author} {\bibfnamefont {R.}~\bibnamefont
  {{San-Jos{\'e}}}},\ }\href {https://doi.org/10.48550/arXiv.2601.21514}
  {\bibinfo {title} {Transversal gates for quantum {{CSS}} codes}} (\bibinfo
  {year} {2026}),\ \Eprint {https://arxiv.org/abs/2601.21514} {arXiv:2601.21514
  [cs]} \BibitemShut {NoStop}%
\bibitem [{\citenamefont {Tansuwannont}\ \emph {et~al.}(2025)\citenamefont
  {Tansuwannont}, \citenamefont {Takada},\ and\ \citenamefont
  {Fujii}}]{Tansuwannont:2025riy}%
  \BibitemOpen
  \bibfield  {author} {\bibinfo {author} {\bibfnamefont {T.}~\bibnamefont
  {Tansuwannont}}, \bibinfo {author} {\bibfnamefont {Y.}~\bibnamefont
  {Takada}},\ and\ \bibinfo {author} {\bibfnamefont {K.}~\bibnamefont
  {Fujii}},\ }\href {https://doi.org/10.48550/arXiv.2503.19790} {\bibinfo
  {title} {Clifford gates with logical transversality for self-dual {{CSS}}
  codes}} (\bibinfo {year} {2025}),\ \Eprint {https://arxiv.org/abs/2503.19790}
  {arXiv:2503.19790 [quant-ph]} \BibitemShut {NoStop}%
\bibitem [{\citenamefont {Tansuwannont}\ \emph {et~al.}(2026)\citenamefont
  {Tansuwannont}, \citenamefont {Chan},\ and\ \citenamefont
  {Takagi}}]{Tansuwannont:2026bmw}%
  \BibitemOpen
  \bibfield  {author} {\bibinfo {author} {\bibfnamefont {T.}~\bibnamefont
  {Tansuwannont}}, \bibinfo {author} {\bibfnamefont {T.}~\bibnamefont {Chan}},\
  and\ \bibinfo {author} {\bibfnamefont {R.}~\bibnamefont {Takagi}},\ }\href
  {https://doi.org/10.48550/arXiv.2602.09788} {\bibinfo {title} {Construction
  of the full logical {{Clifford}} group for high-rate quantum {{Reed-Muller}}
  codes using only transversal and fold-transversal gates}} (\bibinfo {year}
  {2026}),\ \Eprint {https://arxiv.org/abs/2602.09788} {arXiv:2602.09788
  [quant-ph]} \BibitemShut {NoStop}%
\bibitem [{\citenamefont {Webster}\ \emph {et~al.}(2023)\citenamefont
  {Webster}, \citenamefont {Quintavalle},\ and\ \citenamefont
  {Bartlett}}]{Webster:2023cpy}%
  \BibitemOpen
  \bibfield  {author} {\bibinfo {author} {\bibfnamefont {M.~A.}\ \bibnamefont
  {Webster}}, \bibinfo {author} {\bibfnamefont {A.~O.}\ \bibnamefont
  {Quintavalle}},\ and\ \bibinfo {author} {\bibfnamefont {S.~D.}\ \bibnamefont
  {Bartlett}},\ }\href {https://doi.org/10.1088/1367-2630/acfc5f} {\bibfield
  {journal} {\bibinfo  {journal} {New Journal of Physics}\ }\textbf {\bibinfo
  {volume} {25}},\ \bibinfo {pages} {103018} (\bibinfo {year} {2023})},\
  \Eprint {https://arxiv.org/abs/2303.15615} {arXiv:2303.15615 [quant-ph]}
  \BibitemShut {NoStop}%
\bibitem [{\citenamefont {Webster}\ \emph {et~al.}(2022)\citenamefont
  {Webster}, \citenamefont {Brown},\ and\ \citenamefont
  {Bartlett}}]{Webster:2022kdn}%
  \BibitemOpen
  \bibfield  {author} {\bibinfo {author} {\bibfnamefont {M.~A.}\ \bibnamefont
  {Webster}}, \bibinfo {author} {\bibfnamefont {B.~J.}\ \bibnamefont {Brown}},\
  and\ \bibinfo {author} {\bibfnamefont {S.~D.}\ \bibnamefont {Bartlett}},\
  }\href {https://doi.org/10.22331/q-2022-09-22-815} {\bibfield  {journal}
  {\bibinfo  {journal} {Quantum}\ }\textbf {\bibinfo {volume} {6}},\ \bibinfo
  {pages} {815} (\bibinfo {year} {2022})},\ \Eprint
  {https://arxiv.org/abs/2203.00103} {arXiv:2203.00103 [quant-ph]} \BibitemShut
  {NoStop}%
\bibitem [{\citenamefont {Calderbank}\ and\ \citenamefont
  {Shor}(1996)}]{Calderbank:1995dw}%
  \BibitemOpen
  \bibfield  {author} {\bibinfo {author} {\bibfnamefont {A.~R.}\ \bibnamefont
  {Calderbank}}\ and\ \bibinfo {author} {\bibfnamefont {P.~W.}\ \bibnamefont
  {Shor}},\ }\href {https://doi.org/10.1103/PhysRevA.54.1098} {\bibfield
  {journal} {\bibinfo  {journal} {Physical Review A: Atomic, Molecular, and
  Optical Physics}\ }\textbf {\bibinfo {volume} {54}},\ \bibinfo {pages} {1098}
  (\bibinfo {year} {1996})},\ \Eprint {https://arxiv.org/abs/quant-ph/9512032}
  {arXiv:quant-ph/9512032} \BibitemShut {NoStop}%
\bibitem [{\citenamefont {Steane}(1996)}]{Steane:1996ghp}%
  \BibitemOpen
  \bibfield  {author} {\bibinfo {author} {\bibfnamefont {A.~M.}\ \bibnamefont
  {Steane}},\ }\href {https://doi.org/10.1103/PhysRevLett.77.793} {\bibfield
  {journal} {\bibinfo  {journal} {Physical Review Letters}\ }\textbf {\bibinfo
  {volume} {77}},\ \bibinfo {pages} {793} (\bibinfo {year} {1996})}\BibitemShut
  {NoStop}%
\bibitem [{\citenamefont {Terhal}(2015)}]{Terhal:2013vbm}%
  \BibitemOpen
  \bibfield  {author} {\bibinfo {author} {\bibfnamefont {B.~M.}\ \bibnamefont
  {Terhal}},\ }\href {https://doi.org/10.1103/RevModPhys.87.307} {\bibfield
  {journal} {\bibinfo  {journal} {Reviews of Modern Physics}\ }\textbf
  {\bibinfo {volume} {87}},\ \bibinfo {pages} {307} (\bibinfo {year} {2015})},\
  \Eprint {https://arxiv.org/abs/1302.3428} {arXiv:1302.3428 [quant-ph]}
  \BibitemShut {NoStop}%
\bibitem [{\citenamefont {Bombin}\ and\ \citenamefont
  {{Martin-Delgado}}(2007)}]{Bombin:2006cd}%
  \BibitemOpen
  \bibfield  {author} {\bibinfo {author} {\bibfnamefont {H.}~\bibnamefont
  {Bombin}}\ and\ \bibinfo {author} {\bibfnamefont {M.~A.}\ \bibnamefont
  {{Martin-Delgado}}},\ }\href {https://doi.org/10.1063/1.2731356} {\bibfield
  {journal} {\bibinfo  {journal} {Journal of Mathematical Physics}\ }\textbf
  {\bibinfo {volume} {48}},\ \bibinfo {pages} {052105} (\bibinfo {year}
  {2007})},\ \Eprint {https://arxiv.org/abs/quant-ph/0605094}
  {arXiv:quant-ph/0605094} \BibitemShut {NoStop}%
\bibitem [{\citenamefont {Kitaev}(1997)}]{Kitaev:1997quantum}%
  \BibitemOpen
  \bibfield  {author} {\bibinfo {author} {\bibfnamefont {A.}~\bibnamefont
  {Kitaev}},\ }\href@noop {} {\bibfield  {journal} {\bibinfo  {journal}
  {Russian Mathematical Surveys}\ }\textbf {\bibinfo {volume} {52}},\ \bibinfo
  {pages} {1191} (\bibinfo {year} {1997})}\BibitemShut {NoStop}%
\bibitem [{\citenamefont {Haruna}(2026)}]{Haruna:2025piy}%
  \BibitemOpen
  \bibfield  {author} {\bibinfo {author} {\bibfnamefont {J.}~\bibnamefont
  {Haruna}},\ }\href {https://doi.org/10.1093/ptep/ptag105} {\bibfield
  {journal} {\bibinfo  {journal} {Progress of Theoretical and Experimental
  Physics}\ }\textbf {\bibinfo {volume} {2026}},\ \bibinfo {pages} {073B12}
  (\bibinfo {year} {2026})},\ \Eprint {https://arxiv.org/abs/2511.15224}
  {arXiv:2511.15224 [hep-th]} \BibitemShut {NoStop}%
\bibitem [{\citenamefont {Gottesman}\ and\ \citenamefont
  {Chuang}(1999)}]{Gottesman:1999tea}%
  \BibitemOpen
  \bibfield  {author} {\bibinfo {author} {\bibfnamefont {D.}~\bibnamefont
  {Gottesman}}\ and\ \bibinfo {author} {\bibfnamefont {I.~L.}\ \bibnamefont
  {Chuang}},\ }\href {https://doi.org/10.1038/46503} {\bibfield  {journal}
  {\bibinfo  {journal} {Nature}\ }\textbf {\bibinfo {volume} {402}},\ \bibinfo
  {pages} {390} (\bibinfo {year} {1999})},\ \Eprint
  {https://arxiv.org/abs/quant-ph/9908010} {arXiv:quant-ph/9908010}
  \BibitemShut {NoStop}%
\bibitem [{\citenamefont {Feng}\ \emph {et~al.}(2026)\citenamefont {Feng},
  \citenamefont {Xue}, \citenamefont {Kobayashi}, \citenamefont {Hsin},\ and\
  \citenamefont {Chen}}]{Feng:2025bww}%
  \BibitemOpen
  \bibfield  {author} {\bibinfo {author} {\bibfnamefont {Y.}~\bibnamefont
  {Feng}}, \bibinfo {author} {\bibfnamefont {H.}~\bibnamefont {Xue}}, \bibinfo
  {author} {\bibfnamefont {R.}~\bibnamefont {Kobayashi}}, \bibinfo {author}
  {\bibfnamefont {P.-S.}\ \bibnamefont {Hsin}},\ and\ \bibinfo {author}
  {\bibfnamefont {Y.-A.}\ \bibnamefont {Chen}},\ }\href
  {https://doi.org/10.48550/arXiv.2601.00064} {\bibinfo {title} {Pauli
  stabilizer formalism for topological quantum field theories and generalized
  statistics}} (\bibinfo {year} {2026}),\ \Eprint
  {https://arxiv.org/abs/2601.00064} {arXiv:2601.00064 [quant-ph]} \BibitemShut
  {NoStop}%
\bibitem [{\citenamefont {Barkeshli}\ \emph {et~al.}(2024)\citenamefont
  {Barkeshli}, \citenamefont {Chen}, \citenamefont {Hsin},\ and\ \citenamefont
  {Kobayashi}}]{Barkeshli:2022edm}%
  \BibitemOpen
  \bibfield  {author} {\bibinfo {author} {\bibfnamefont {M.}~\bibnamefont
  {Barkeshli}}, \bibinfo {author} {\bibfnamefont {Y.-A.}\ \bibnamefont {Chen}},
  \bibinfo {author} {\bibfnamefont {P.-S.}\ \bibnamefont {Hsin}},\ and\
  \bibinfo {author} {\bibfnamefont {R.}~\bibnamefont {Kobayashi}},\ }\href
  {https://doi.org/10.21468/SciPostPhys.16.4.089} {\bibfield  {journal}
  {\bibinfo  {journal} {SciPost Physics}\ }\textbf {\bibinfo {volume} {16}},\
  \bibinfo {pages} {089} (\bibinfo {year} {2024})},\ \Eprint
  {https://arxiv.org/abs/2211.11764} {arXiv:2211.11764 [cond-mat]} \BibitemShut
  {NoStop}%
\bibitem [{\citenamefont {Breuckmann}\ \emph {et~al.}(2026)\citenamefont
  {Breuckmann}, \citenamefont {Davydova}, \citenamefont {Eberhardt},\ and\
  \citenamefont {Tantivasadakarn}}]{Breuckmann:2024axh}%
  \BibitemOpen
  \bibfield  {author} {\bibinfo {author} {\bibfnamefont {N.~P.}\ \bibnamefont
  {Breuckmann}}, \bibinfo {author} {\bibfnamefont {M.}~\bibnamefont
  {Davydova}}, \bibinfo {author} {\bibfnamefont {J.~N.}\ \bibnamefont
  {Eberhardt}},\ and\ \bibinfo {author} {\bibfnamefont {N.}~\bibnamefont
  {Tantivasadakarn}},\ }\href {https://doi.org/10.1007/s00220-026-05570-z}
  {\bibfield  {journal} {\bibinfo  {journal} {Communications in Mathematical
  Physics}\ }\textbf {\bibinfo {volume} {407}},\ \bibinfo {pages} {86}
  (\bibinfo {year} {2026})},\ \Eprint {https://arxiv.org/abs/2410.16250}
  {arXiv:2410.16250 [quant-ph]} \BibitemShut {NoStop}%
\bibitem [{\citenamefont {Bauer}(2025)}]{Bauer:2024qpc}%
  \BibitemOpen
  \bibfield  {author} {\bibinfo {author} {\bibfnamefont {A.}~\bibnamefont
  {Bauer}},\ }\href {https://doi.org/10.22331/q-2025-03-25-1673} {\bibfield
  {journal} {\bibinfo  {journal} {Quantum}\ }\textbf {\bibinfo {volume} {9}},\
  \bibinfo {pages} {1673} (\bibinfo {year} {2025})},\ \Eprint
  {https://arxiv.org/abs/2403.12119} {arXiv:2403.12119 [quant-ph]} \BibitemShut
  {NoStop}%
\end{thebibliography}%

\appendix
\Crefname{section}{Appendix}{Appendices}

\section{Properties of the Hadamard product of free modules}
\label[appendix]{app:hadamard_product_vector_space}

Let $V$, $W$, $U$ be submodules of $R^n$ for a commutative ring $R$.
Recall that $V \otimes_H W \coloneqq \spn_R\{v \circ w : v \in V,\, w \in W\}$, where $(v \circ w)_i = v_i w_i$.

\begin{proposition}
\label{prop:hadamard_properties}
The Hadamard product module satisfies:
\begin{enumerate}[label=(\roman*)]
\item $V \otimes_H W$ is a submodule of $R^n$;
\item $V \otimes_H W = W \otimes_H V$;
\item $V \otimes_H (W_1 + W_2) = V \otimes_H W_1 + V \otimes_H W_2$ and $(\alpha V) \otimes_H W = \alpha(V \otimes_H W)$ for all $\alpha \in R$;
\item $(V \otimes_H W) \otimes_H U = V \otimes_H (W \otimes_H U)$.
\end{enumerate}
\end{proposition}

\begin{proof}
All properties follow from the componentwise definition $(v \circ w)_i = v_i w_i$ and the corresponding properties of multiplication in $R$: (i) $V \otimes_H W$ is an $R$-linear span, hence a submodule; (ii) the commutativity of $R$ gives $v_i w_i = w_i v_i$; (iii) the distributivity and scalar compatibility of $R$ pass to the span; (iv) the associativity of $R$ gives $(v_i w_i)u_i = v_i(w_i u_i)$.
\end{proof}

Next, we give a special property of the iterated Hadamard product of free modules over $\mathbb{Z}_{2}$.
The following proposition states that it forms an expanding sequence of submodules.
\begin{proposition}
\label{prop:iterated_hadamard_product_Z2}
Let $V$ be a submodule of $\mathbb{Z}_2^n$. Then the iterated Hadamard product of $V$ forms an expanding sequence of submodules:
\begin{align}
V^{\otimes_H k} \subset V^{\otimes_H (k+1)},
\end{align}
for any $k \geq 1$.
\end{proposition}

\begin{proof}
For any $v \in V^{\otimes_H k}$, we can write $v$ as a linear combination of vectors of the form $v_1 \circ v_2 \circ \cdots \circ v_k$ with $v_i \in V$ as
\begin{align}
  v = \sum_j \alpha_j (v_{1,j} \circ v_{2,j} \circ \cdots \circ v_{k,j}),
\end{align}
where $\alpha_j \in \mathbb{Z}_2$.

Because each $v_{i,j}$ is a $\mathbb{Z}_2$-vector, it is invariant under the Hadamard product with itself:
\begin{align}
  v_{i,j} \circ v_{i,j} = v_{i,j}.
\end{align}
Using this property (for $i=k$ here), we can rewrite each term in the sum as
\begin{align}
  v = \sum_j \alpha_j (v_{1,j} \circ v_{2,j} \circ \cdots \circ v_{k,j} \circ v_{k,j}).
\end{align}
Because $v_{k,j} \in V$, we have $v_{1,j} \circ v_{2,j} \circ \cdots \circ v_{k,j} \circ v_{k,j} \in V^{\otimes_H (k+1)}$.
Thus, $v \in V^{\otimes_H (k+1)}$ and we conclude that $V^{\otimes_H k} \subset V^{\otimes_H (k+1)}$.
\end{proof}

Finally, we give the proof of the well-definedness of the Hadamard product of a quotient module and a submodule.
\begin{proposition}
\label{prop:hadamard_product_quotient_submodule}
Let $V,W,U \subset R^n$ be submodules and $W \subset V$. Then the Hadamard product of the quotient module $V/W$ and the submodule $U$ is well defined as a submodule of $R^n$:
\begin{align}
  (V/W) \otimes_H U \subset R^n.
\end{align}
\end{proposition}

\begin{proof}
  If $v' = v + w$ for some $w \in W$, then
  \begin{align}
    [v'] \circ u & = [v' \circ u]_{V \otimes_H W} \\
    & = [(v + w) \circ u]_{V \otimes_H W} \\
    & = [v \circ u + w \circ u]_{V \otimes_H W} \\
    & = [v \circ u]_{V \otimes_H W} \\
     & = [v] \circ u,
  \end{align}
  where we have used the fact that $w \circ u \in W \otimes_H U$ because $u \in U$ and $w \in W$.
  So the product does not depend on the choice of representative $v$ in the equivalence class $[v]$, and thus is well defined.
\end{proof}

\section{Duality between quotients of submodules}
\label[appendix]{app:duality_quotient_submodule}

In this section, we establish a duality between quotients of submodules in $\mathbb{Z}_{2^m}^n$ under the standard nondegenerate inner product modulo $2^m$.

\begin{lemma}[Duality of quotients]
\label{lem:dual_quotient}
For finitely generated submodules $V \subset W \subset \mathbb{Z}_{2^m}^n$ with some integer $m \geq 1$, there is a natural isomorphism between the quotient $W/V$ and the homomorphism module $\Hom_{\mathbb{Z}_{2^m}}(V^\perp/W^\perp,\mathbb{Z}_{2^m})$:
\begin{align}
  W/V \cong \Hom_{\mathbb{Z}_{2^m}}(V^\perp/W^\perp,\mathbb{Z}_{2^m}).
\end{align}
In addition, there exists a (noncanonical) isomorphism between a finite $\mathbb{Z}_{2^m}$-module and its dual, so
\begin{align}
  V^\perp/W^\perp \cong W/V.
\end{align}
\end{lemma}

\begin{proof}
Define a map $\Phi$ from the quotient to the homomorphism module,
\begin{align}
\Phi: W/V \to \Hom_{\mathbb{Z}_{2^m}}(V^\perp/W^\perp,\mathbb{Z}_{2^m}),
\end{align}
  by the value of $\Phi([w])$ on $[x] \in V^\perp/W^\perp$ as
\begin{align}
\Phi([w])([x]) \coloneqq \inp{w}{x} ,
\quad
w\in W,\ x\in V^\perp .
\end{align}

We first check that $\Phi$ is well defined.
If $w' = w+v$ with $v\in V$, then for any $x\in V^\perp$,
\begin{align}
\inp{w'}{x} = \inp{w}{x} + \inp{v}{x} = \inp{w}{x},
\end{align}
because $x\in V^\perp$.
Hence $\Phi([w])$ does not depend on the choice of representative $w$ of $[w]\in W/V$.
Similarly, if $x' = x+y$ with $y\in W^\perp$, then
\begin{align}
\inp{w}{x'} = \inp{w}{x} + \inp{w}{y} = \inp{w}{x},
\end{align}
because $w\in W$ and $y\in W^\perp$.
Thus $\Phi([w])$ depends only on the class $[x]\in V^\perp/W^\perp$.

Next, we show injectivity.
Suppose $\Phi([w])=0$ in $\Hom_{\mathbb{Z}_{2^m}}(V^\perp/W^\perp,\mathbb{Z}_{2^m})$, which means that 
\begin{align}
\Phi([w])([x]) = \inp{w}{x} = 0 \quad \text{for all } x\in V^\perp,
\end{align}
so $w\in (V^\perp)^\perp$.
Because the standard inner product on $\mathbb{Z}_{2^m}^n$ is nondegenerate, we have $(V^\perp)^\perp = V$.
Hence $w\in V$, so $[w]=0$ in $W/V$.
Therefore $\Phi$ is injective.

Finally, we show surjectivity.
Let
\begin{align}
\lambda \in \Hom_{\mathbb{Z}_{2^m}}(V^\perp/W^\perp,\mathbb{Z}_{2^m}).
\end{align}
Composing with the quotient map $V^\perp \to V^\perp/W^\perp$, we obtain a homomorphism
\begin{align}
\tilde\lambda: V^\perp \to \mathbb{Z}_{2^m}
\end{align}
that vanishes on $W^\perp$.

Because $\mathbb{Z}_{2^m}$ is a self-injective ring, the homomorphism $\tilde\lambda$ extends from the submodule $V^\perp \subset \mathbb{Z}_{2^m}^n$ to a homomorphism
\begin{align}
\Lambda: \mathbb{Z}_{2^m}^n \to \mathbb{Z}_{2^m}.
\end{align}
By nondegeneracy of the standard inner product, every linear functional on $\mathbb{Z}_{2^m}^n$ is of the form
\begin{align}
\Lambda(\cdot)=\langle w,\cdot\rangle
\end{align}
for a unique $w\in \mathbb{Z}_{2^m}^n$.
Because $\tilde\lambda$ vanishes on $W^\perp$, we have
\begin{align}
\langle w,y\rangle = 0
\quad
\text{for all } y\in W^\perp,
\end{align}
so $w\in (W^\perp)^\perp = W$.
Therefore $[w]\in W/V$ is well defined, and for every $[x]\in V^\perp/W^\perp$,
\begin{align}
\Phi([w])([x])=\langle w,x\rangle=\tilde\lambda(x)=\lambda([x]).
\end{align}
Thus $\Phi([w])=\lambda$, proving surjectivity.
Hence $\Phi$ is an isomorphism:
\begin{align}
  W/V \cong \Hom_{\mathbb{Z}_{2^m}}(V^\perp/W^\perp,\mathbb{Z}_{2^m}),
\end{align}
which is the first claim of the lemma.

In addition, every finite $R=\mathbb Z_{2^m}$-module is noncanonically isomorphic to its $R$-dual.
By the structure theorem based on Smith normal form, we can choose a decomposition
\begin{align}
  V^\perp/W^\perp \cong \bigoplus_i R/(2^{k_i}).
\end{align}
Using the standard identifications
\begin{align}
  R/(2^{k_i}) \cong \Hom_R(R/(2^{k_i}),R),
\end{align}
we obtain a noncanonical isomorphism
\begin{align}
  V^\perp/W^\perp \cong W/V,
\end{align}
which completes the proof.
\end{proof}

\section{Definition of the obstruction maps}
\label[appendix]{app:definition_of_obstruction_maps}

Here, we give a detailed definition of the obstruction maps $\beta_1^{(m)}$ and $\beta_2^{(m)}$.
\begin{definition}[First obstruction map $\beta_1^{(m)}$]
\label{def:beta1_detailed}
Let $[\theta]$ be an element of $V_{LD}^{(m)}$ and $\theta$ be a representative of $[\theta]$ in $\mathbb{Z}_{2^m}^n$.
The first obstruction map $\beta_1^{(m)}$ is defined as a map
\begin{align}
  \beta_1^{(m)}: V_{LD}^{(m)} \to T^{(m)}_1,
\end{align}
where the value of $\beta_1^{(m)}$ on $[\theta]$ is given by
\begin{align}
  \beta_1^{(m)}([\theta]) \coloneqq \left[(\widetilde{H}^{(m)} \odot K)\theta\right]
\end{align}
as an element of the quotient
\begin{align}
  \label{eq:beta1_target_space}
  T^{(m)}_1 \coloneqq \mathbb{Z}_{2^m}^{(n-m_z)r_m}/ \im\left((\widetilde{H}^{(m)} \odot K) \widetilde{G}^{(m)}\right),
\end{align}
where $r_m$ is the number of rows of $\widetilde{H}^{(m)}$.
\end{definition}

\begin{remark}
Because the dimension of $\ker H_Z$ is $n - \rank{H_Z} = n - m_z$, the matrix $K$ has $n-m_z$ rows, and therefore $\widetilde{H}^{(m)} \odot K$ has $(n-m_z)r_m$ rows.
\end{remark}

\begin{remark}
This obstruction map $\beta_1^{(m)}([\theta])$ is well defined, that is, independent of the choice of the representative $\theta$ of $[\theta]$ in $V_{LD}^{(m)}$, which can be shown as follows.
Suppose $\theta' = \theta + \widetilde{G}^{(m)} \alpha'$ for some $\alpha' \in \mathbb{Z}_{2^m}^{d_{LI}^{(m)}}$ is another representative of the same class in $V_{LD}^{(m)}$.
Then the change in $\beta_1^{(m)}$ is given by
\begin{align}
\beta_1^{(m)}([\theta']) - \beta_1^{(m)}([\theta]) = \left[(\widetilde{H}^{(m)} \odot K) \widetilde{G}^{(m)} \alpha'\right],
\end{align}
which is an element of the image of $(\widetilde{H}^{(m)} \odot K) \widetilde{G}^{(m)}$.
Therefore, the change corresponds to a trivial element in the quotient, and $\beta_1^{(m)}([\theta])$ is invariant under the change of representative.
\end{remark}

Next, we give a detailed definition of the second obstruction map $\beta_2^{(m)}$.
\begin{definition}[Second obstruction map $\beta_2^{(m)}$]
\label{def:beta2_detailed}
Let $\theta$ be a representative of a class in $V_{LD}^{(m)}$ such that $\beta_1^{(m)}([\theta]) = 0$.
The second obstruction map $\beta_2^{(m)}$ is defined as
\begin{align}
  \beta_2^{(m)}([\theta]) \coloneqq \left[\frac{1}{2^m}\widetilde{H}^{(m)} \bigl(\theta + \widetilde{G}^{(m)} \alpha_0\bigr)\right],
\end{align}
where $\alpha_0$ is any solution to~\Cref{eq:lift_condition_split_1}, and the equivalence class is taken in
\begin{align}
  \label{eq:beta2_target_space}
  T^{(m)}_2 \coloneqq \frac{\mathbb{Z}_2^{r_m}}{\eval{\im \frac{\widetilde{H}^{(m)} \widetilde{G}^{(m)}}{2^m}}_{\ker((\widetilde{H}^{(m)} \odot K) \widetilde{G}^{(m)})} + \im \widetilde{H}^{(m)}},
\end{align}
where $r_m$ is the number of rows of $\widetilde{H}^{(m)}$.
\end{definition}

\begin{remark}
This definition is independent of the choice of $\alpha_0$, which can be shown as follows.
Suppose $\alpha_0'$ is another solution to~\Cref{eq:lift_condition_split_1}, so that $\alpha_0' = \alpha_0 + \eta'$ for some $\eta' \in \ker((\widetilde{H}^{(m)} \odot K) \widetilde{G}^{(m)})$.
Then the change in the numerator of $\beta_2^{(m)}([\theta])$ is given by
\begin{align}
\frac{1}{2^m}\widetilde{H}^{(m)}
\widetilde{G}^{(m)} (\alpha_0' - \alpha_0)
=
\frac{1}{2^m}\widetilde{H}^{(m)} \widetilde{G}^{(m)} \eta', 
\end{align}
which is an element of the image of $\widetilde{H}^{(m)} \widetilde{G}^{(m)}/2^m$ restricted to the kernel of $(\widetilde{H}^{(m)} \odot K) \widetilde{G}^{(m)}$.
Therefore, the change corresponds to a trivial element in the quotient, and $\beta_2^{(m)}([\theta])$ is well defined independent of the choice of $\alpha_0$.
\end{remark}

\begin{remark}
Moreover, $\beta_2^{(m)}([\theta])$ is independent of the choice of representative $\theta$ of $[\theta]$ in $V_{LD}^{(m)}$.
Suppose $\theta' = \theta + \widetilde{G}^{(m)} \alpha'$ is another representative of the same class in $V_{LD}^{(m)}$.
Because $\beta_1^{(m)}$ is well defined on $V_{LD}^{(m)}$, we have $\beta_1^{(m)}([\theta']) = 0$ as well.
Set $\alpha_0' \coloneqq \alpha_0 - \alpha'$; then
\begin{align}
(\widetilde{H}^{(m)} \odot K)(\theta' + \widetilde{G}^{(m)} \alpha_0') 
& = (\widetilde{H}^{(m)} \odot K)(\theta + \widetilde{G}^{(m)} \alpha_0)
\\ & = 0 \pmod{2^m},
\end{align}
so $\alpha_0'$ is a valid solution to~\Cref{eq:lift_condition_split_1} for $\theta'$.
Using this choice,
\begin{align}
\frac{1}{2^m}\widetilde{H}^{(m)}(\theta' + \widetilde{G}^{(m)} \alpha_0')
=
\frac{1}{2^m}\widetilde{H}^{(m)}(\theta + \widetilde{G}^{(m)} \alpha_0),
\end{align}
so $\beta_2^{(m)}([\theta']) = \beta_2^{(m)}([\theta])$ in the quotient.
Therefore, $\beta_2^{(m)}$ is a well-defined map on $V_{LD}^{(m)}$.
\end{remark}

\end{document}